\documentclass[12pt,reqno,notitlepage]{amsart}

\usepackage[letterpaper,margin=1.0in]{geometry}
\usepackage{anyfontsize}
\usepackage[T1]{fontenc}

\usepackage{amsmath,amssymb,amsthm}

\usepackage{graphicx}
\usepackage[svgnames]{xcolor}
\usepackage{pgf}
\usepackage{tikz}
\usetikzlibrary{decorations.pathreplacing, positioning, arrows.meta}
\usepackage{mathrsfs}
\usepackage{stmaryrd}
\SetSymbolFont{stmry}{bold}{U}{stmry}{m}{n}
\usepackage{layout}
\usepackage{mathtools}
\usepackage{paralist}
\usepackage{multirow}
\usepackage{comment}
\usepackage{rotating}
\usepackage{chngcntr}
\usepackage{apptools}
\usepackage{mathpazo}

\usepackage{xspace}
\usepackage{bm}
\usepackage{setspace}
\usepackage{accents}
\usepackage{enumitem}
\usepackage{footnote}
\usepackage{tabularx}
\usepackage{threeparttable}

\makesavenoteenv{tabular}
\makesavenoteenv{table}

\usepackage[colorlinks=true, linkcolor=Maroon, urlcolor=Maroon, citecolor=Maroon]{hyperref}

\usepackage[authoryear]{natbib}

\theoremstyle{plain}
\newtheorem{assumption}{Assumption}

\newtheorem{theorem}{Theorem}

\newtheorem{lemma}{Lemma}%[section]

\theoremstyle{remark}

\usepackage{todonotes}

\renewcommand{\Pr}{\mathbb{P}}

\usepackage{subfig}
\newcommand{\rtheta}{\theta^{(r)}}
\newcommand{\one}{\mathbb{1}}

\usepackage{centernot}

\usepackage{cleveref}
\crefname{figure}{figure}{figures}
\creflabelformat{figure}{#2#1#3}

\crefname{equation}{equation}{equations}
\crefrangeformat{equation}{(#3#1#4) to~(#5#2#6)}
\creflabelformat{equation}{(#2#1#3)}

\crefname{lemma}{lemma}{lemmas}
\creflabelformat{lemma}{#2#1#3}
\crefname{proposition}{proposition}{propositions}
\creflabelformat{proposition}{#2#1#3}

\crefname{corollary}{corollary}{corollaries}
\creflabelformat{corollary}{#2#1#3}

\crefname{condition}{condition}{conditions}
\creflabelformat{condition}{#2#1#3}
\crefname{assumption}{assumption}{assumptions}
\creflabelformat{assumption}{#2#1#3}
\crefname{remark}{remark}{remarks}
\creflabelformat{remark}{#2#1#3}
\crefname{appendix}{appendix}{appendices}
\creflabelformat{appendix}{#2#1#3}

\allowdisplaybreaks
\begin{document}

\renewcommand{\thefootnote}{\fnsymbol{footnote}}

\begin{center}
\Large\textsc{Bounding the Effect of Persuasion with Monotonicity Assumptions: Reassessing the Impact of TV Debates}\footnote{We are grateful to Vincent Pons, the editor, and three anonymous referees  for their encouraging and insightful comments.}
\vspace{1ex}
\end{center}

\vspace{.3cm}
\begin{center}
    \begin{tabular}{ccc}
        \large\textsc{Sung Jae Jun}\footnote{Department of Economics, suj14@psu.edu.} 
        &
        &
        \large\textsc{Sokbae Lee}\footnote{Department of Economics, sl3841@columbia.edu.}
        \\
        Pennsylvania State University
        &
        &
        Columbia University
    \end{tabular}
\end{center}

\begin{center}
\vspace{0.3cm} 
    February 2026 
\vspace{0.3cm}
\end{center}

\noindent
\textbf{Abstract.}
%Economists have long emphasized the powerful effects of persuasive communication, yet recent evidence by \citet{TVdebate:23} suggests that televised debates may not sway voters as much as commonly believed. In this paper, we revisit their findings through the lens of persuasion rates, introduced by \citet{dellavigna2007fox} and extended by \citet{jun2023identifying,jun2024aprt}, focusing on monotonicity assumptions to partially identify these effects without relying on exogenous treatment, instruments, or difference-in-differences. Our re-analysis reaffirms that the upper bound on the persuasive effects of TV debates is rather small. 
%
%
%Televised debates of presidential candidates are often considered the exemplar of persuasive communication, but \citet[][LP23]{TVdebate:23} have recently documented 
%%by using a linear event-study regression model 
%that they may not sway voters as much as folk beliefs. In this paper, we revisit their findings through the lens of the persuasion rate. For this purpose, we propose a novel robust framework that does not rely on exogenous treatment, parallel trends, or availability of credible instruments. Instead, we exploit monotonicity assumptions that look natural in LP23's design even without needing to condition on any covariates. We then partially identify the persuasion rate and other related parameters. Our robust approach reaffirms that the sharp upper bound on the persuasive effects of TV debates is indeed rather small.  
%

Televised debates between presidential candidates are often regarded as the exemplar of persuasive communication. Yet, recent evidence from \citet{TVdebate:23} indicates that they may not sway voters as strongly as popular belief suggests. We revisit their findings through the lens of the persuasion rate and introduce a robust framework that does not require exogenous treatment, parallel trends, or credible instruments. Instead, we leverage plausible monotonicity assumptions to partially identify the persuasion rate and related parameters. Our results reaffirm that the sharp upper bounds on the persuasive effects of TV debates remain modest. 
% [shortened to be less than 100 words]

%%%%%%%%%%%%%%%%%%%%%%%%%%%%%%%%%%%%%%%%%%%%%%%%%%%%%%%%%%%


%\todo[inline]{
%The number of words now is 6124, which is below the target number 6200.
%}

%\end{titlepage}

\raggedbottom

\clearpage
%%%%%%%%%%%%%%%%%%%%%%%%%%%%%%%%%%%%%%%%%%%%%%%%%%%%%%%%%%%%%%%%%%%%%%%%%
%%%% Main text entry area:
%%%%%%%%%%%%%%%%%%%%%%%%%%%%%%%%%%%%%%%%%%%%%%%%%%%%%%%%%%%%%%%%%%%%%%%%%

\section{Introduction}\label{section:intro}

Economists have extensively documented empirical evidence demonstrating how persuasive efforts can affect the decisions and actions of consumers, voters, donors, and investors \citep[see][for a survey of the early literature]{dellavigna2010persuasion}. Indeed, the art of persuasion has captivated observers for centuries, from ancient rhetoricians to modern communication experts. In today's media-driven world, few events draw more attention than a live TV debate between U.S. presidential candidates, often billed as a make-or-break moment that can sway millions of undecided voters. However, recent evidence from \citet[][LP hereafter]{TVdebate:23} challenges this common belief. Their large-scale study examines 56 TV debates across 31 elections in seven countries, and they find no discernible impact of those debates on vote choices, beliefs, or policy preferences. 

In this paper, we revisit LP’s findings to reassess the effectiveness of TV debates from a different, robust perspective. Our framework focuses on the persuasion rate and related parameters, yet does not rely on linearity, additive separability, exogenous treatments, parallel trends, or credible instruments. Instead, it draws on a pair of monotonicity assumptions in the spirit of \citet{manski2000monotone}. Because the sharp lower bounds are always trivially zero, we concentrate on the sharp upper bounds on the persuasion rate and other measures. Our results reaffirm LP’s conclusion in that these upper bounds remain small, implying that the overall effects of TV debates appear limited.

The persuasion rate, first introduced by \citet{dellavigna2007fox} to evaluate the efficacy of persuasive efforts, has been extensively used to make a comparison across diverse settings. Building upon this concept, \citet{jun2023identifying} formulate the persuasion rate as a causal parameter that quantifies how a persuasive message affects recipients' behavior, and they conduct a formal identification analysis based on exogenous treatment assignments or credible instruments.  Extending this line of research, \citet{jun2024aprt} propose the forward and backward average persuasion rates on the treated as additional causal parameters of interest, and they develop a difference-in-differences (DID) framework to identify, estimate, and conduct inference for them.
Additional research on persuasion rates includes an analysis of persuasion types among compliers \citep{yu2023binary}, identification of the average persuasion rate under sample selection \citep{possebom2022probability}, an extension of the persuasion effect to continuous decisions \citep{kaji2023assessing},
and covariate-assisted bounds on the persuasion rate without monotonicity assumptions \citep{ji2023model}.

In this paper, we explore implications of two monotonicity assumptions: monotone treatment response \citep[][MTR]{manski1997} and monotone treatment selection \citep[][MTS]{manski2000monotone}. The MTR assumption states that the effect of the informational treatment of interest is directional, and it has also been employed by \citet{jun2023identifying,jun2024aprt} in the context of persuasion rates. In our application to TV debates, MTR implies that exposure to the debate never increases uncertainty regarding vote intentions: i.e., debates provide information that helps voters crystallize their preferences. For identifying the persuasion rate, MTR links the joint probability of two potential outcomes with the average treatment effect. The MTS assumption is new in studying the persuasion rate. In our context, it requires that the probability of vote choice consistency (voting exactly as intended at the time of the pre-election survey) not decrease as the survey date approaches election day, which seems natural.

Since the foundational work of \citet{manski1997} and \citet{manski2000monotone}, employing  MTR, MTS, and other types of monotonicity assumptions has been fruitful in the literature on partial identification. To sample a few, see
\citet{Kreider2007JASA},
\citet{vytlacil2007dummy},
%\citet{bhattacharya2008treatment,BSV-2012},
\citet{BSV-2012},
%\citet{Manski:Pepper:2009},
\citet{Kreider-et-al:2012},
%\citet{Gundersen2012Joe},
\citet{okumura2014concave},
\citet{Kim-et-al:2018},
\citet{MSV-2019},
%\citet{Jensen2023JoHE},
and
\citet{Jun2024JBES}
among others.
Our paper is another example that shows the usefulness of the MTR and MTS assumptions.

It is worth noting that using MTR and MTS does not alter the estimands one would compute under MTR and exogenous treatments. What changes is how we interpret these estimands: the population quantities that point-identify the persuasion rate and its reverse version under MTR with exogenous treatments become the sharp upper bounds on these same measures when the exogenous treatment assumption is replaced with MTS. We illustrate these findings by reanalyzing LP’s data.

This robust interpretation of our estimands can be extended further. Both the persuasion rate and its reverse version are defined by conditioning on a counterfactual event. However, \citet{pearl1999probabilities,pearl2009causality,Pearl2015causes} has advocated conditioning on an observable event rather than a counterfactual one. We show that this conceptual difference is inconsequential in our approach: our estimands serve as sharp upper bounds not only on the persuasion rate and its reverse version, but also on the probabilities of causation proposed by Pearl, when we solely rely on MTR and MTS for identification. 

Our empirical findings suggest that watching TV debates has, at most, a modest impact on vote choice consistency: the estimated average treatment effect is no higher than 0.7\%, while the persuasion rate is about 3\% at most. This implies that relatively few voters switch from an inconsistent to a consistent vote choice by watching debates. Our results align with those of LP, even though our methodology is quite different.
We also remark that the upper bound on the persuasion rate is larger than that on the average treatment effect because, under the MTR assumption, the persuasion rate is derived by rescaling the average treatment effect. The rescaling factor specifically avoids ``preaching to the converted'' by conditioning on relevant events. However, even by this alternative measure, the effect of TV debates remains limited.

The rest of the paper is organized as follows. In \Cref{sec:prelim}, we introduce the causal persuasion rate and related parameters, and discuss the MTR assumption, which will be used throughout. \Cref{sec:id} focuses on identification, highlighting that the estimands that point-identify the parameters of interest under exogenous treatments become sharp upper bounds on those parameters when the MTS assumption is imposed. In \Cref{sec:pc}, we consider \citet{pearl1999probabilities}’s probabilities of causation and show that the same sharp upper bounds apply under MTR and MTS. In \Cref{sec:eg}, we apply our methodology to revisit LP’s data. In \Cref{sec:conclusions}, we provide concluding remarks. Finally, in the appendix, we give the proofs of our main theoretical results.

\section{Preliminaries}\label{sec:prelim}

For $d\in \{0,1\}$, let $Y(d)$ denote a binary potential outcome, $D$ the observed binary treatment, and therefore, the observed outcome is $Y = DY(1) + (1-D) Y(0)$. 
Because our identification results do not rely on covariates, we focus on the case without them. Alternatively, one can view our assumptions and identification results as ``implicitly conditional on covariates.'' The causal parameters we consider are: 
\begin{align*} 
\theta =  \Pr\{ Y(1) = 1\mid Y(0) = 0 \} 
\ \ \text{ and } \ \
\rtheta = \Pr\{ Y(0) = 0 \mid Y(1) = 1 \}.
\end{align*}
The conditional probability $\theta$ is known as the average persuasion rate (APR), whereas $\rtheta$ is its reverse version (R-APR). Here, $\theta$ measures the proportion of individuals, among those who would not have taken the action of interest in the absence the treatment, who would have changed their behavior in the presence of the treatment. The reverse version $\rtheta$ focuses on individuals who would have taken the action of interest with the treatment, and it answers how many of them would not have done so without the treatment. See 
\citet{dellavigna2007fox}, \citet{dellavigna2010persuasion}, and 
\citet{jun2023identifying,jun2024aprt} for more discussion on these and related parameters.

We start with making two assumptions. 

\begin{assumption}[Regularity] \label{ass:overlap}
For all $(y,d) \in \{0,1\}^2$, $0 < \Pr\{ Y = y, D= d\} < 1$.
%For all $(y,d,d') \in \{0,1\}^3$, $0 < \Pr\{ Y(d) = y, D= d'\} < 1$.
\end{assumption}

\begin{assumption}[Monotone Treatment Response (MTR)]\label{ass:mtr}
We have $Y(1) \geq Y(0)$ with probability one. 
\end{assumption}

\Cref{ass:overlap} imposes weak regularity conditions on the observed data. \Cref{ass:mtr} has been used for partial identification \citep[e.g.][]{manski1997,manski2000monotone,Jun2024JBES} or in the context of informational treatments \citep{jun2023identifying,jun2024aprt}. In the latter case, it is often interpreted as an assumption that says, ``Persuasive messages are directional, and there is no backlash.''

In our application, \Cref{ass:mtr} is motivated by LP’s research design, in which respondents are surveyed twice: first to elicit pre-election voting intentions, and later to record their final post-election choices. The treatment indicator $D$ equals $1$ if the initial survey occurs after a TV debate and $0$ otherwise. The outcome variable $Y$, referred to as \emph{vote choice consistency}, equals $1$ if a respondent’s pre-election intention matches their eventual vote, and $0$ otherwise; notably, $Y=0$ whenever the initial response is ``undecided.'' Formally, let $Y(1)$ and $Y(0)$ represent vote choice consistency with and without a debate, respectively. Panel A of \Cref{fig:combined_methodology} illustrates the coding of these potential outcomes.

In this context, \Cref{ass:mtr} means that a debate may help clarify a voter's preference but should not confuse a previously decided voter. For instance, Scenario 1 in Panel A depicts a respondent who is undecided without the debate ($Y(0)=0$) but forms a consistent preference for Candidate A after watching it ($Y(1)=1$), a case that satisfies MTR. In contrast, Scenario 2 illustrates the reverse case, where a respondent who would have been consistent without the debate becomes undecided (and thus inconsistent) after watching it. We consider this ``confusion'' scenario implausible if debates serve to provide information and clarify preferences; therefore, MTR excludes this possibility. Finally, Scenarios 3 and 4 show examples where $Y(0) = Y(1) = 0$, which are consistent with MTR but yield no variation in the outcome.

It is worth noting that the plausibility of the MTR assumption depends heavily on the definitions of the treatment and outcome variables. For instance, \citet{cavgias2024media} analyze the 1989 Brazilian election, where the treatment is defined as having \emph{access} to a biased TV debate. In that context, the debate, which was widely regarded as biased towards the right-wing candidate, could persuade a voter to switch candidates (e.g., from Lula to Collor). Consequently, a treated voter would appear inconsistent ($Y(1)=0$) relative to their pre-debate baseline, while an untreated voter would remain consistent ($Y(0)=1$). This creates a violation of MTR.
In contrast, our treatment is defined by the \emph{timing} of the survey. If a debate persuades a voter to switch candidates, a respondent surveyed \emph{after} the debate ($D=1$) will report their new preference, which matches their final vote ($Y(1)=1$). Meanwhile, if surveyed \emph{before} the debate ($D=0$), they report their old preference, which contradicts their final vote ($Y(0)=0$). Thus, in our design, persuasion actually reinforces the MTR condition ($Y(1) \geq Y(0)$).  Therefore, we view ``the chance of the debate causing confusion about their own preferences'' (Scenario 2) as a primary threat to our identification strategy.

Under \Cref{ass:overlap,ass:mtr}, the population consists of three (nonempty) groups: 
(i) \emph{never-persuadable} (NP): $Y(1) = 0,  Y(0) = 0$;
(ii) \emph{already-persuaded} (AP): $Y(1) = 1,  Y(0) = 1$;
and
(iii) \emph{treatment-persuadable} (TP): $Y(1) = 1,  Y(0) = 0$.
As we never observe the pair $\{Y(1),  Y(0) \}$ jointly, we cannot assign a type to each observational unit; however, we can aim to identify the share of the three groups. Using the group probabilities, we can rewrite the parameters of interest as 
\begin{align}\label{PR:as:shares} 
\theta =  \frac{\Pr\{ \mathrm{TP} \}}{\Pr\{ \mathrm{TP} \} + \Pr\{ \mathrm{NP} \}}
\ \ \text{ and } \ \
\rtheta = \frac{\Pr\{ \mathrm{TP} \}}{\Pr\{ \mathrm{TP} \} + \Pr\{ \mathrm{AP} \}}.
\end{align}
We will now consider conditions under which we can point or partially identify the causal parameters. 

\section{Identification}\label{sec:id}

We first consider the standard case of exogenous treatment as a benchmark. 
\begin{assumption}[Exogenous Treatment] \label{ass:unconf}
For $d\in\{0,1\}$, $Y(d)$ is independent of $D$.
\end{assumption} 
\Cref{ass:overlap} ensures that $\Pr(Y=1\mid D=1) > 0$ and $\Pr(Y=1\mid D=0) < 1$, and we define the following population quantities: 
\begin{align*} 
\theta_U &= \frac{\Pr( Y = 1 \mid D=1) - \Pr(Y=1\mid D=0)}{\Pr(Y=0\mid D=0)}, \\
\rtheta_U &=  \frac{\Pr( Y = 1 \mid D=1) - \Pr(Y=1\mid D=0)}{\Pr(Y=1\mid D=1)},
\end{align*} 
which are directly identified from the data. We now state the following benchmark result; see, e.g., \citet[Theorem~1]{jun2023identifying}.

\begin{lemma}[Benchmark]\label{lem:benchmark}
Under \Cref{ass:overlap,ass:unconf,ass:mtr}, we have
\[
\Pr\{ \mathrm{NP}  \} =  \Pr(Y=0\mid D=1)
\quad \text{and}\quad 
\Pr\{ \mathrm{AP}  \} =  \Pr(Y=1\mid D=0).
\] 
In particular, $\theta$ and $\rtheta$ are point-identified by $\theta_U$ and $\rtheta_U$, respectively.
%\proof
%The proof is standard and hence omitted.  \qed
\end{lemma}

\Cref{ass:unconf} is popular, but it can be restrictive, as it precludes endogenous treatment assignment. Also, we are implicit here, but \Cref{ass:unconf} is often stated after conditioning on a sufficiently large vector of covariates, which can be too demanding in practice. To go beyond \Cref{ass:unconf}, we now consider an alternative assumption that allows endogenous treatment, namely monotone treatment selection.

\begin{assumption}[Monotone Treatment Selection (MTS)]\label{ass:mts}
For all $d\in\{0,1\}$, we have 
\[
   \Pr\{ Y(d) = 1\mid D=1 \} \geq \Pr\{ Y(d) = 1\mid D=0 \}.
\]  
\end{assumption}

In LP's study, the key identifying assumption is that, once potential confounders are accounted for, the date of the debate is uncorrelated with vote choice consistency.  Therefore, roughly speaking, their event study design is based on (conditional) exogenous treatment assignment. However, controlling for all possible confounders in an observational study is inevitably challenging, even though LP makes a convincing case. 

As a robust alternative, we employ the MTS assumption introduced by \citet{manski2000monotone}. Formally, it requires that those in the treatment group be no less likely to ``succeed'' than those in the control group. Even if $D$ is a pure choice variable, \Cref{ass:mts} stays valid as long as those who chose $D=1$ are not strictly inferior on average to those who opted out.  In the context of LP, being surveyed later in the election cycle (i.e., \(D=1\)) naturally places respondents closer to the election day, and hence it makes vote consistency more likely, rendering MTS a plausible condition even without needing to control for covariates. See \Cref{sec:eg} for more details.

\begin{theorem} \label{thm:main}
Suppose that \Cref{ass:overlap,ass:mtr,ass:mts} hold. The sharp bounds on $\Pr\{ \mathrm{NP} \}$ and $\Pr\{ \mathrm{AP} \}$ are given by
\begin{align*}
\Pr\{ \mathrm{NP} \} \in \left[ \Pr( Y = 0 \mid D=1), \Pr( Y = 0) \right]
\quad \text{and}\quad 
\Pr\{ \mathrm{AP} \} \in \left[ \Pr( Y = 1 \mid D=0), \Pr( Y = 1) \right],
%\\
%\Pr\{ \mathrm{TP} \} &\in \left[ 0, \Pr( Y = 1\mid D=1) - \Pr(Y=1\mid D=0) \right].
\end{align*}
where $\Pr\{ \mathrm{TP}\} = 1- \Pr\{ \mathrm{NP}\} - \Pr\{ \mathrm{AP}\}$.
In particular, the sharp identified intervals on $\theta, \rtheta$ are given by 
\begin{align*} 
0 \leq \theta \leq \theta_U 
\quad \text{ and } \quad 
0 \leq \rtheta \leq \rtheta_U.
\end{align*} 
\proof 
See the appendix. \qed 
\end{theorem} 

\Cref{thm:main} shows that, for both $\theta$ and $\rtheta$, the same population quantities $\theta_U$ and $\rtheta_U$ arise whether we assume \Cref{ass:unconf} or \Cref{ass:mts}, although their interpretation changes. For example, under \Cref{ass:overlap,ass:unconf,ass:mtr}, $\theta$ is point-identified by $\theta_U$, while that same $\theta_U$ becomes the sharp upper bound on $\theta$ under \Cref{ass:overlap,ass:mtr,ass:mts}. Consequently, a plug-in estimator of  $\theta_U$ is consistent for $\theta$ under exogeneity, and it will be consistent for the sharp upper bound on $\theta$ under the weaker assumption of MTS.

\citet{dellavigna2007fox} and \citet{dellavigna2010persuasion} propose the following quantity as the persuasion rate:
\[
\theta_{DK}
=
\frac{\Pr[Y=1\mid Z=1] - \Pr[Y=1\mid Z=0]}{\Pr[D=1\mid Z=1] - \Pr[D=1\mid Z=0]}
\times \frac{1}{\Pr[Y=0\mid Z=0]},
\]
where $Z$ is a binary instrument. Under full compliance, that is, when $D=Z$ almost surely, this quantity coincides with the upper bound on the persuasion rate derived in \Cref{thm:main}. More generally, under partial compliance, when $D \neq Z$, it has been pointed out that $\theta_{DK}$ can be misleading and may inflate the causal persuasion rate \citep{jun2023identifying}.

In addition, \Cref{thm:main} establishes bounds on the shares of the NP, AP, and TP groups. In particular, $\Pr( Y = 0 \mid D=1)$ and $\Pr( Y = 1 \mid D=0)$ are now interpreted as sharp lower bounds on $\Pr\{ \mathrm{NP} \}$ and $\Pr\{ \mathrm{AP} \}$, respectively. Further, the proof of \Cref{thm:main} shows that \Cref{ass:mtr,ass:mts} yield
\[
0 \leq \Pr\{Y(1) = 1\} - \Pr\{ Y(0) = 1\} \leq \Pr( Y = 1 \mid D=1) - \Pr(Y=1\mid D=0),
\]
where the bounds are sharp. The sharp bounds on the share of TP follow from here because $\Pr\{\mathrm{TP}\} = \Pr\{Y(1) = 1\} - \Pr\{ Y(0) = 1\}$ equals the average treatment effect (ATE) under MTR.

%Since $\Pr\{ \mathrm{TP} \} = 1- \Pr\{ \mathrm{NP}\} - \Pr\{\mathrm{AP}\} = \Pr\{Y(1) = 1\} - \Pr\{ Y(0) = 1\}$, the sharp bounds on $\Pr\{ \mathrm{TP} \}$ are given by 
%\[
%\Pr\{ \mathrm{TP} \} 
%\in 
%\left[ 0,\  \Pr( Y = 1 \mid D=1) - \Pr( Y = 1 \mid D=0) \right].
%\]

\section{Estimation and Inference}

We now describe our approach to estimation and inference for the parameters of interest. Since inference for $\theta$ and $\rtheta$ follows the same principle, we will focus on the former. 

\subsection{Persuasion Rates}

We obtain plug-in estimators for $\theta_U$ by using the coefficients from an OLS regression of $Y$ on $D$ with an intercept; e.g., the coefficient of $D$ corresponds to the numerator of $\theta_U$.  If \Cref{ass:unconf} is maintained, standard two-sided confidence intervals are applicable because $\theta_U$ point-identifies $\theta$ in this case (see \Cref{lem:benchmark}). Under the weaker \Cref{ass:mts}, however, $[0, \theta_U]$ is the sharp identified bounds on $\theta$, making one-sided inference appropriate.

Specifically, under \Cref{ass:mts}, the theoretical lower bound is known to be zero. Accordingly, we construct the $(1-\alpha)$ confidence interval as
\[
\mathrm{CI} = \left[ 0, \ \hat{\theta}_U + c_\alpha \cdot \widehat{\mathrm{SE}}_U \right],
\]
where $\hat{\theta}_U$ is the estimator for the upper bound, $\widehat{\mathrm{SE}}_U$ is the corresponding standard error (computed via the delta method), and $c_\alpha = \Phi^{-1}(1-\alpha)$ is the standard one-sided critical value (e.g., $1.645$ for $\alpha=0.05$). %Confidence intervals {\color{blue}for ATE} and $\rtheta$ are constructed analogously.

Notably, this confidence interval is valid for both the true parameter $\theta$ and the identified set $[0, \theta_U]$. While general partial identification methods typically require larger critical values to cover the set \citep{Imbens/Manski:04}, our setting is simpler because the lower bound is fixed at zero. Consequently, the coverage probability for the identified set coincides with that of the parameter at the least favorable configuration (i.e., $\theta = \theta_U$). Thus, the standard one-sided critical value $c_\alpha$ ensures asymptotic coverage of at least $1-\alpha$ for both targets.

Furthermore, uniformity issues are simpler than in the general cases discussed in e.g., \citet{canay2017practical} too. Since $\Pr(\theta \in \mathrm{CI}) \geq \Pr\{ \theta_U(P) \in \mathrm{CI} \}$ for any data-generating process $P\in \mathcal{P}$ with the corresponding $\theta_U(P) \geq \theta \geq 0$, the coverage properties of $\mathrm{CI}$ will be uniform in $P\in \mathcal{P}$, provided that the normal approximation of $\{\hat{\theta}_U - \theta_U(P)\}/\widehat{\mathrm{SE}}_U$ is uniform. See the online appendices for more details.

\subsection{Specification Testing}

The identification result under \Cref{ass:mtr,ass:mts} implies that the upper bound must be non-negative ($\theta_U \geq 0$). This restriction is testable. Following the logic of \citet{Stoye:07}, we can construct a confidence set that becomes empty if the data provide strong evidence against the validity of the identification assumptions. Specifically, we define the modified confidence set $\mathrm{CI}_{\mathrm{mod}}$ as:
\[
\mathrm{CI}_{\mathrm{mod}} :=
\begin{cases}
[0, \ \hat\theta_U + c_\alpha \widehat{\mathrm{SE}}_U] & \text{if } \hat\theta_U + c_\alpha \widehat{\mathrm{SE}}_U \geq 0, \\
\emptyset & \text{otherwise}.
\end{cases}
\]
%This confidence set is empty if and only if the estimated upper bound is significantly negative. %(specifically, when the $t$-statistic $\hat\theta_U / \widehat{\mathrm{SE}}_U < -c_\alpha$). 
The event $\mathrm{CI}_{\mathrm{mod}} = \emptyset$ corresponds to rejecting $H_0: \theta_U \geq 0$ in favor of $H_1: \theta_U < 0$ at significance level $\alpha$, which indicates that MTR or MTS is violated.

\subsection{Shares of Already-Persuaded and Never-Persuadable Types}
As before, if \Cref{ass:unconf} is maintained, standard two-sided confidence intervals are applicable because the proportions of AP and NP types are point-identified (see \Cref{lem:benchmark}). Under the weaker \Cref{ass:mts}, however, the situation differs from that of persuasion rates. Here, the sharp identified sets require the bounds at both ends be estimated (see \Cref{thm:main}). To construct confidence intervals for these parameters, we adopt the method of \citet{Imbens/Manski:04} and \citet{Stoye:07}, which adjusts the critical value to account for the estimation of the length of the identified set.

Let $[\theta_{s,L}, \theta_{s,U}]$ denote the identified set for type $s \in \{\mathrm{AP}, \mathrm{NP}\}$. The $(1-\alpha)$ confidence interval is given by
\[
\mathrm{CI}_s = \left[ \hat\theta_{s,L} - c_\alpha \widehat{\mathrm{SE}}_{s,L}, \quad \hat\theta_{s,U} + c_\alpha \widehat{\mathrm{SE}}_{s,U} \right],
\]
where $c_\alpha$ is chosen to satisfy
\[
\Phi\left( c_\alpha + \frac{\hat{\Delta}_s}{\max(\widehat{\mathrm{SE}}_{s,L}, \widehat{\mathrm{SE}}_{s,U})} \right) - \Phi(-c_\alpha) = 1-\alpha,
\]
with $\hat{\Delta}_s$ representing the estimated length of the interval. This approach ensures (uniformly) valid coverage for the true parameter $\theta_s$ while mitigating the conservatism inherent in intervals designed to cover the entire identified set. We refer to \citet{Stoye:07} for formal details.

\section{Bounding the Probabilities of Causation}\label{sec:pc}
In this section, we discuss a few related parameters that have been discussed in the causal inference literature. Both $\theta$ and $\rtheta$ are conditioned on a counterfactual outcome, which is not directly observed. Advocating the idea that we should consider probabilities of counterfactual events given the information that is observed, \citet{pearl1999probabilities} introduces the following parameters and calls them probabilities of causation: 
\begin{align*} 
\mathrm{PS} &=  \Pr\{ Y(1) = 1\mid Y = 0, D=0 \},\\ 
\mathrm{PN} &=  \Pr\{ Y(0) = 0\mid Y = 1, D=1 \}, \\
\mathrm{PNS}  &=  \Pr\{ Y(1) = 1,  Y(0) = 0 \}, 
\end{align*}
where PS, PN, and PNS, respectively, denote the probability of sufficiency, the probability of necessity, and the probability of necessity and sufficiency.
For further discussion, see, e.g., \citet{pearl1999probabilities,Yamamoto:2012,Dawid2014fitting,Dawid2022effects,Ding:2024:arXiv:prob_necessity}. 
These quantities can also be expressed as 
\begin{align*} 
\mathrm{PS} &=  \frac{\Pr\{ \mathrm{TP} , D=0 \}}{\Pr\{ \mathrm{TP} , D=0 \} + \Pr\{ \mathrm{NP} , D=0 \}},\\
\mathrm{PN} &=  \frac{\Pr\{ \mathrm{TP} , D=1 \}}{\Pr\{ \mathrm{TP} , D=1 \} + \Pr\{ \mathrm{AP} , D=1 \}},\\
\mathrm{PNS}  &=  \Pr\{ \mathrm{TP} \}.
\end{align*}

Under \Cref{ass:overlap,ass:unconf,ass:mtr}, we have that $\theta = \mathrm{PS}$, $\rtheta = \mathrm{PN}$, and 
$\mathrm{PNS} = \Pr( Y = 1 \mid D=1) - \Pr(Y=1\mid D=0)$, as initially noted by \citet{pearl1999probabilities}.
More interestingly, under \Cref{ass:overlap,ass:mtr,ass:mts}, we have the following results.

\begin{theorem} \label{thm:main:pc}
Suppose that \Cref{ass:overlap,ass:mtr,ass:mts} hold. Then, the sharp identified intervals on $\mathrm{PS}$,  $\mathrm{PN}$, and $\mathrm{PNS}$, respectively, are given by 
\begin{align*} 
0 \leq \mathrm{PS} \leq \theta_U, \;
0 \leq \mathrm{PN} \leq \rtheta_U, \; \text{ and } \;
0 \leq \mathrm{PNS} \leq \Pr( Y = 1 \mid D=1) - \Pr( Y = 1 \mid D=0),
\end{align*} 
\proof 
See the appendix. \qed 
\end{theorem} 
The sharp bounds on $\textrm{PNS}$ are immediate from \Cref{thm:main} because $\Pr\{\mathrm{TP}\}$ equals ATE under MTR. The main point here is that the sharp identified bounds on $\mathrm{PS}$ and $\mathrm{PN}$ under \Cref{ass:overlap,ass:mtr,ass:mts} coincide with those on $\theta$ and $\rtheta$, respectively.

%\Cref{thm:main:pc} therefore shows that under  \Cref{ass:overlap,ass:mtr,ass:mts}, the bounds on $\mathrm{PS}$ and $\mathrm{PN}$ coincide with those on $\theta$ and $\rtheta$, respectively. 
%In other words, $\theta_U$ and $\rtheta_U$ also serve as the sharp upper bounds on $\mathrm{PS}$ and $\mathrm{PN}$, respectively.
%In addition, it confirms that $\mathrm{PNS}$ shares the same bounds as $\Pr\{ \mathrm{TP} \}$, since both correspond to the average treatment effect.

\section{Bounds under Alternative Assumptions}

In this section, we derive the identified regions for the persuasion rate $\theta$ and the reverse persuasion rate $\theta^{(r)}$ under relaxed assumptions. We consider three scenarios: (i) no monotonicity assumptions; (ii) MTR only; and (iii) MTS only.

\begin{theorem}\label{thm:alternative_bounds}
Let \Cref{ass:overlap} hold. The sharp identified intervals for $\theta$ and $\theta^{(r)}$, respectively, are given as follows: 
\begin{enumerate}
    \item With no other assumptions: $\theta \in [0, 1]$ and $\theta^{(r)} \in [0, 1]$.
    \item With \Cref{ass:mtr} only: $\theta \in [0, 1]$ and $\theta^{(r)} \in [0, 1]$.
    \item With \Cref{ass:mts} only: $\theta \in [0, \overline{\theta}_{\text{MTS}}]$ and $\theta^{(r)} \in [0, \overline{\theta}^{(r)}_{\text{MTS}}]$, where
    \begin{align*}
        \overline{\theta}_{\text{MTS}}
        &:= \min\left\{ \frac{\Pr(Y=1\mid D=1)}{\Pr(Y=1,D=1)+\Pr(Y=0,D=0)},\ 1 \right\}, \\
        \overline{\theta}^{(r)}_{\text{MTS}}
        &:= \min\left\{ \frac{\Pr(Y=0\mid D=0)}{\Pr(Y=1,D=1)+\Pr(Y=0,D=0)},\ 1 \right\}.
    \end{align*}
\end{enumerate}
\end{theorem}

\begin{proof}
See the online appendix. The argument closely parallels the proof of \Cref{thm:main}.
\end{proof}

\Cref{thm:alternative_bounds} shows that the MTR assumption alone provides no identifying power for either $\theta$ or $\theta^{(r)}$, whereas MTS alone may yield non-trivial bounds. Intuitively, the role of MTR is restricted to relating the joint probabilities of the potential outcomes to their marginal distributions, whereas MTS has identifying power for the marginals of the potential outcomes themselves.

Furthermore, note that
\begin{align*}
 \overline{\theta}_{\text{MTS}} < 1
%& \Longleftrightarrow
%\frac{\Pr(Y=1\mid D=1)}{\Pr(Y=1,D=1)+\Pr(Y=0,D=0)} < 1 \\
& \Longleftrightarrow
\Pr(Y=1,D=1) < \Pr(D=1)\{\Pr(Y=1,D=1)+\Pr(Y=0,D=0)\} \\
& \Longleftrightarrow
\Pr(Y=0,D=0) > \Pr(D=0)\{\Pr(Y=1,D=1)+\Pr(Y=0,D=0)\} \\
& \Longleftrightarrow
%\frac{\Pr(Y=0\mid D=0)}{\Pr(Y=1,D=1)+\Pr(Y=0,D=0)} > 1.
\overline{\theta}^{(r)}_{\text{MTS}} = 1.
\end{align*}
Thus, whenever MTS yields a non-trivial upper bound on $\theta$, the sharp bound on $\rtheta$ becomes trivial. 
%the ratio governing the upper bound for $\theta^{(r)}$ exceeds one. This implies that $\overline{\theta}^{(r)}_{\text{MTS}} = 1$, which is a trivial bound.
%
In contrast, we have 
\[
\overline{\theta}_{\text{MTS}} \geq \theta_U
\quad\text{and}\quad
\overline{\theta}^{(r)}_{\text{MTS}} \geq \theta^{(r)}_U
\]
by construction, where both $\theta_U$ and $\theta^{(r)}_U$ can be non-trivial. We therefore adopt the joint assumption of MTR and MTS as our preferred identifying restrictions.

We now turn to the sharp identified regions for $\mathrm{PS}$, $\mathrm{PN}$, and $\mathrm{PNS}$ under the three scenarios. Define
\begin{align*}
\overline{\mathrm{PNS}}_{\text{NA}} &:= \Pr(Y=1,D=1) + \Pr(Y=0,D=0), \\
\overline{\mathrm{PNS}}_{\text{MTS}} &:= \min\{ \Pr(Y=0\mid D=0),\ \Pr(Y=1\mid D=1) \}, \\
\overline{\mathrm{PS}}_{\text{MTS}} &:= \min\left\{ \frac{\Pr(Y=1\mid D=1)}{\Pr(Y=0\mid D=0)},\ 1 \right\}, \\
\overline{\mathrm{PN}}_{\text{MTS}} &:= \min\left\{ \frac{\Pr(Y=0\mid D=0)}{\Pr(Y=1\mid D=1)},\ 1 \right\}.
\end{align*}

\begin{theorem}\label{thm:alternative_bounds_pearl_tian}
Let \Cref{ass:overlap} hold. The sharp identified intervals for $\mathrm{PS}$, $\mathrm{PN}$, and $\mathrm{PNS}$, respectively, are given as follows: 
\begin{enumerate}
    \item With no other assumptions: $\mathrm{PS} \in [0, 1]$, $\mathrm{PN} \in [0, 1]$, and $\mathrm{PNS} \in [0, \overline{\mathrm{PNS}}_{\text{NA}}]$.
    \item With \Cref{ass:mtr} only: $\mathrm{PS} \in [0, 1]$, $\mathrm{PN} \in [0, 1]$, and $\mathrm{PNS} \in [0, \overline{\mathrm{PNS}}_{\text{NA}}]$.
    \item With \Cref{ass:mts} only: $\mathrm{PS} \in [0, \overline{\mathrm{PS}}_{\text{MTS}}]$, $\mathrm{PN} \in [0, \overline{\mathrm{PN}}_{\text{MTS}}]$, and $\mathrm{PNS} \in [0, \overline{\mathrm{PNS}}_{\text{MTS}}]$.
\end{enumerate}
\end{theorem}

\begin{proof}
Case (i) is given in \citet[Section 4.2.1]{tian2000probabilities}, and case (ii) is given in \citet[Section 4.4.1]{tian2000probabilities}. For case (iii), see the online appendix. The argument closely parallels the proof of \Cref{thm:main:pc}.
\end{proof}

As in the case of $\theta$ and $\rtheta$, the MTR assumption alone does not have any identification power for any of the probabilities of causation. Our results under MTS only, i.e., \Cref{thm:alternative_bounds_pearl_tian}(iii), relate to the bounds derived by \citet[Theorem~1]{tian2000probabilities} under exogeneity. %\Cref{ass:unconf} is equivalent to \citet{tian2000probabilities}'s definition of exogeneity (i.e., Definition 12), and Theorem 1 of that work states the associated bounds. 
Comparing the two sets of the results, we find that the upper bounds yielded by MTS alone are identical to those under exogeneity. The additional identification power provided by exogeneity applies exclusively to the lower bounds.

Finally, \Cref{tab:bounds_summary} summarizes the hierarchy of bounds for the persuasion rates and the probabilities of causation under different combinations of assumptions.

\section{Empirical Example}\label{sec:eg}

%In this section, we revisit the findings of LP to reevaluate how effectively TV debates persuade (undecided) voters. The debates examined in LP took place between 5 and 44 days (on average 24 days) before the election, precisely when vote choice consistency tends to increase most rapidly. By employing an event-study design and capitalizing on the variation in debate timing across different elections, LP control for the number of days preceding each election using daily fixed effects to distinguish the impact of the debates from general time trends. Interestingly, LP find no significant differences in how voters' choices evolved before versus after debate periods, rejecting any effect larger than $0.5$ percentage points. This null result holds across diverse voter types, different election contexts, and even tight, uncertain races.

%In their analysis, LP employ an event study design using 331,000 respondent-debate-election observations to estimate the short-term dynamics of debate exposure. Their specification isolates the impact of the debate one to three days post-event relative to the broadcast day, controlling for secular time trends relative to the election, sociodemographic covariates, and debate-by-election fixed effects among others. Their identifying assumption is that the specific timing of the debate is exogenous conditional on these controls (see \Cref{appnx:LP:reg} for the full econometric specification).

In this section, we revisit the findings of LP to reevaluate the effectiveness of TV debates in persuading (undecided) voters. The debates examined in LP took place between 5 and 44 days (on average 24 days) before the election, precisely when vote choice consistency tends to increase most rapidly. Utilizing an event-study design based on 331,000 respondent-debate-election observations, LP isolate the impact of debates occurring one to three days post-event relative to the broadcast day. Their specification controls for secular time trends, sociodemographic covariates, and debate-by-election fixed effects, assuming the specific timing of the debate is exogenous conditional on these factors. See \Cref{appnx:LP:reg} for the full econometric specification.  Interestingly, LP find no significant differences in how voters' choices evolved before versus after debate periods, rejecting any effect larger than $0.5$ percentage points. This null result holds across diverse voter types, different election contexts, and even tight, uncertain races.

In our empirical analysis, we adopt an alternative strategy. Rather than controlling for the extensive set of variables used in LP, we restrict the sample to a narrow window of $[-3,+3]$ days relative to the debate. This localized approach focuses on the period immediately surrounding the event. We define the pre-debate days as the control group ($D=0$) and the debate day together with the post-debate days as the treatment group ($D=1$).  Panel B of \Cref{fig:combined_methodology} provides a visualization of the treatment definition ($D$) and the sample window.

Crucially, our method does not require conditioning on additional covariates. Unlike standard exogeneity assumptions, the MTR and MTS assumptions in this setting do not rely on covariate adjustment for their validity. MTR is imposed at the individual level, while MTS will be motivated without controlling for covariates. Therefore, our approach avoids reliance on functional form restrictions or strict exogeneity assumptions.

First, consider the plausibility of MTR, for which it is helpful to recall that the outcome variable is vote choice consistency defined by the agreement between respondents' reported voting intentions prior to the election and their realized vote choices reported after the election; if the voter is undecided before the election, then the outcome is coded as inconsistent voting whatever the final vote is. Panel A of \Cref{fig:combined_methodology} explains how the potential outcomes would be coded in several scenarios. The idea is that because debates provide information, the likelihood of a voter remaining undecided should decrease in the presence of a debate. For example, imagine a voter who would vote for Candidate A before a debate, would become disappointed and undecided after watching the debate, and would end up just not voting at all as a result. If this voter is observed before the debate, then $Y(0) = 0$ will be observed, whereas if she is surveyed after the debate, then $Y(1) = 0$ will be observed. Therefore, MTR is still satisfied in this case. From this perspective, information revealed during a debate (including details about candidates' traits such as health or competence) helps voters refine and finalize their preferences prior to voting. Even if such information leads some voters to revise their initial intentions, this does not violate MTR, provided that debate exposure does not induce voters to become ``confused'' about their own preferences. 

Next, regarding the plausibility of MTS, respondents surveyed later in the electoral cycle are closer to the election day and thus more likely to exhibit consistent voting behavior on average. This supports the unconditional MTS assumption. Further, our use of a $\pm$3-day estimation window around the debate adds extra support for its plausibility.  While respondent composition might shift over a long campaign (e.g., urban voters being surveyed earlier than rural voters), such structural shifts are implausible within such a short period. We verify this by showing that observable covariates are balanced between the pre- and post-debate groups. Specifically, when we regress the treatment indicator on the sociodemographic controls used in the paper (gender, age, income quartiles, employment status, and education), the resulting p-value for the joint F-test of significance (with standard errors clustered at the debate level) is 0.445. This failure to reject the null hypothesis supports the validity of the unconditional MTS assumption in our local setting.

\Cref{tab:combined_results} presents the estimation results divided into two panels. Panel A reports the estimates and upper confidence bounds of (the sharp upper bounds on) the Average Treatment Effect (ATE), the Average Persuasion Rate (APR, denoted as $\theta$), and the Reverse Average Persuasion Rate (R-APR, denoted as $\rtheta$). In the first row, labeled ``All,'' we consider the full sample ($N=42,136$) and find that the estimated impact of TV debates is small. The point estimate of the upper bound on ATE is 0.68 percentage points (standard error 0.56), with a 95\% upper confidence bound (UCB) of 1.60\%. These results align closely with those in LP, although we rely on a notably different methodology. Since the share of Treatment-Persuadable (TP) respondents is identical to ATE in our setting, these small values imply that very few voters' consistency relies exclusively on watching TV debates.

The middle three columns in Panel A report the upper bounds on APR, its standard error (SE), and its 95\% UCB, while the last three columns present analogous estimates for R-APR. Recall that these rates rescale the share of TP relative to different denominators:
\[
\textrm{APR} = \frac{\textrm{share of TP}}{\textrm{share of TP} + \textrm{share of NP}},
\quad 
\textrm{R-APR} = \frac{\textrm{share of TP}}{\textrm{share of TP} + \textrm{share of AP}}.
\]
Because both denominators are less than one, APR and R-APR are mechanically no less than ATE. Specifically, APR is estimated at 3.37\% with a UCB of 7.91\%, while R-APR is estimated at 0.84\% with a UCB of 1.98\%. The APR estimate exhibits greater amplification than that of R-APR because the share of Never-Persuadable (NP) voters is significantly smaller than that of Already-Persuaded (AP) voters. This will be confirmed in Panel B of \Cref{tab:combined_results}.

Economically, these measures answer different counterfactual questions. APR considers those who would \textit{not} have been consistent without the debate and asks what fraction would change their behavior if they watched it. In contrast, R-APR considers those who \textit{were} consistent after watching the debate and asks what fraction would have differed without it. The former effect is small, but the latter is even smaller. We note that the reported 95\% confidence intervals are one-sided because the lower bounds are trivially zero, so only the upper bounds require estimation (see \Cref{thm:main}).

Subsequent rows of Panel A of \Cref{tab:combined_results} reveal notable heterogeneity across demographic groups and countries. Younger voters (Age $< 50$) appear more susceptible to persuasion than older voters: the estimated sharp upper bound on ATE for the under-50 group is 1.06\% (UCB 2.45\%), compared to just 0.40\% (UCB 1.36\%) for the over-50 group. Similarly, APR for younger viewers is more than double that of older viewers (4.83\% versus 2.13\%). When disaggregated by country, the results for the U.S.\ and the U.K.\ are qualitatively similar, though the U.S.\ sample shows a slightly higher persuasion rate: APR in the U.S.\ is estimated at 5.09\% (UCB 12.42\%), compared to 3.69\% (UCB 12.63\%) in the U.K. However, the standard errors for these subsamples are larger, indicating less precision than in the pooled analysis.

Panel B of \Cref{tab:combined_results} provides insight into the shares of different types: the vast majority of the population belongs to the AP or NP categories. For the full sample, the estimated proportion of AP types is bounded between 79.88\% and 80.25\%, while the proportion of NP types lies between 19.44\% and 19.75\%. The tight identification region confirms that approximately 80\% of the viewing population was already committed to the outcome prior to the debate. This ``ceiling effect'' is more pronounced in the U.S.\ subsample, where the lower bound on the share of the AP type is 87.30\%, slightly higher than the U.K.'s 83.59\%.

Overall, our findings indicate that TV debates exert, at best, a modest influence on vote choice consistency, echoing LP's results. However, when focusing on APR---which excludes AP voters from the relevant population---the potential scope for debate effects is larger. Even so, this evidence remains suggestive because (i) the identified parameter is merely an upper bound on the true persuasive effect, and (ii) the standard errors associated with the APR estimates are relatively large.

\section{Conclusions}\label{sec:conclusions}

We have revisited LP's findings on the effects of TV debates by using a new framework involving the persuasion rate. Instead of relying on exogenous treatments or other strong identifying assumptions, we have leveraged two weak monotonicity assumptions (MTR and MTS) to derive sharp bounds on various measures of persuasive effects, including the persuasion rate and its reverse version. 
Our analysis has shown that the impact of TV debates on vote choice consistency remains modest, which is consistent with LP's conclusion. In addition, we have highlighted that our bounds continue to serve as sharp bounds on Pearl's probabilities of causation under the same monotonicity assumptions, further underscoring the robustness of our approach. Overall, our results suggest that while TV debates may provide helpful information to voters, their influence on electoral choices is limited.

Finally, we summarize a practical roadmap for empirical implementation, as illustrated in Figure \ref{fig:decision_tree}. We consider a baseline setting with binary outcomes, binary treatments, and observed covariates under the MTR assumption. The first step is to assess data completeness. When outcomes or treatments are missing for some units, the sample selection bounds of \citet{possebom2022probability} are applicable under treatment exogeneity. In the absence of sample selection, the appropriate estimator depends on the treatment assignment mechanism. Under exogenous treatment, persuasion rates are point identified by sample analogs, as shown in Lemma~\ref{lem:benchmark}. When treatment is endogenous, identification requires additional structure, such as instrumental variables \citep{jun2023identifying,yu2023binary} or panel data combined with parallel trends assumptions \citep{jun2024aprt}. If neither valid instruments nor panel data are available, the bounds developed in this paper remain informative under the extra assumption of MTS.

An important direction for future research is to broaden the current framework to settings beyond the reach of existing methods, including cases featuring both treatment endogeneity and sample selection, and to develop identification and inference procedures that remain informative under such combined challenges.

%%%%%%%%%%%%%%%%%%%%%%%%%%%%%%%%%%%%%%%%%%%%%%%%%%
%%% Main Appendix 
%%%%%%%%%%%%%%%%%%%%%%%%%%%%%%%%%%%%%%%%%%%%%%%%%%%
%\clearpage
\begin{appendix}  
\section*{Appendix}

\subsection*{Proof of Theorem \ref{thm:main}: }

We first articulate the empirical contents of MTR and MTS, after which we discuss identification of $\Pr\{ \mathrm{NP} \}$, $\Pr\{ \mathrm{TP} \}$, $\Pr\{ \mathrm{AP} \}$, $\theta$, and $\rtheta$.

\subsubsection*{Empirical Contents}

We start from tabulating the joint distribution of $\bigl( Y(0), Y(1), D\bigr)$ as follows: 
\begin{table} [ht]
\begin{tabular}{c|cccc}
$Y(0), Y(1)$   &   $(0,0)$    & $(0,1)$    &   $(1,0)$     &    $(1,1)$   \\
\hline 
$D=0$          &    $q_1$     &  $q_2$     &    $0$        &     $q_3$ \\
$D=1$          &    $q_4$     &  $q_5$     &    $0$        &     $q_6$  \\
\hline 
Prob Restrictions & \multicolumn{4}{c}{$\sum_{j=1}^6 q_j = 1$.}   \\ 
                  & \multicolumn{4}{c}{$0\leq q_j\leq 1$ for $j=1,2,\cdots, 6$.}   \\
\hline 
\end{tabular} 
\end{table}

Here, \Cref{ass:mtr} has already been imposed. \Cref{ass:mts} can be formulated by 
\begin{align} 
    \frac{q_6}{q_4+q_5+q_6}     &\geq  \frac{q_3}{q_1+q_2+q_3} \label{eq:mts0}  \\
    \frac{q_5+q_6}{q_4+q_5+q_6} &\geq \frac{q_2+q_3}{q_1+q_2+q_3}, \label{eq:mts1} 
\end{align} 
where the denominators are not equal to zero by \Cref{ass:overlap}.  \Cref{ass:mtr,ass:mts} do not impose any other restrictions on the $q_j$'s.

For $(y,d) \in \{0,1\}^2$, let $P_{yd} := \Pr(Y = y, D=d) \in (0,1)$, which are directly identified from the data. Then, 
%\begin{align*} 
%q_1 + q_2 &= P_{00},\\ 
%q_3 &= P_{10}, \\
%q_4 &= P_{01}, \\
%q_5+q_6 &= 1- P_{00} - P_{10} - P_{01},
%\end{align*} 
\[
\begin{split}
    q_1 + q_2 &= P_{00},\\
    q_4 &= P_{01},
\end{split}
\qquad \qquad 
\begin{split}
    q_3 &= P_{10}, \\
    q_5+q_6 &= 1- P_{00} - P_{10} - P_{01}.
\end{split}
\]
Therefore, only two of the $q_j$'s are undetermined: i.e., using $q_6 = 1-\sum_{j=1}^5 q_j$, we can write 
\begin{equation} \label{eq:qjs}
q_1 = P_{00} - q_2, \quad 
q_3  = P_{10}, \quad 
q_4 = P_{01}, \quad 
q_6 = P_{11} - q_5,
\end{equation} 
where $0\leq q_2 \leq P_{00}$ and $0\leq q_5\leq P_{11}$.  Using \Cref{eq:qjs}, \Cref{eq:mts0,eq:mts1} can be rewritten as 
\begin{align}  
\frac{P_{11} - q_5}{P_{01}+P_{11}} &\geq \frac{P_{10}}{P_{00}+P_{10}} ,\\
\frac{P_{11}}{P_{01}+P_{11}}   &\geq  \frac{q_2 + P_{10}}{P_{00}+P_{10}}.
\end{align} 
Therefore, the restrictions we have from \Cref{ass:overlap,ass:mtr,ass:mts} are given by \Cref{eq:qjs} with 
\begin{align} 
0\leq q_5 \leq P_{11} - \frac{P_{10}(P_{01}+P_{11})}{P_{00}+P_{10}}, \label{eq:mts0good}\\
0\leq q_2 \leq P_{00} - \frac{P_{01}(P_{10}+ P_{00})}{P_{01}+ P_{11}}. \label{eq:mts1good}
\end{align}

\subsubsection*{Identification of $\Pr\{ \mathrm{NP} \}$, $\Pr\{ \mathrm{TP} \}$, and  $\Pr\{ \mathrm{AP} \}$}

Since
\begin{align*}
\Pr\{ \mathrm{NP} \}  = q_1 + q_4, &\; \Pr\{ \mathrm{TP} \} =q_2 + q_5, \; \Pr\{ \mathrm{AP} \} = q_3 + q_6,
\end{align*}
we have 
\begin{align*}
\Pr\{ \mathrm{NP} \}  = P_{00} + P_{01} - q_2, &\; \Pr\{ \mathrm{TP} \} =q_2 + q_5, \; \Pr\{ \mathrm{AP} \} = P_{10} + P_{11} - q_5. 
\end{align*}
Therefore, the sharp upper (lower) bound on $\Pr\{\mathrm{TP}\}$ can be obtained by maximizing (minimizing) $q_2+q_5$ subject to the constraints in \eqref{eq:mts0good} and \eqref{eq:mts1good}. The other cases are similar. \qed

%Recall that the bounds on $q_2$ and $q_5$ are a rectangle (see \Cref{eq:mts0good,%eq:mts1good}). 
%Thus, the sharp bounds on $\Pr\{ \mathrm{NP} \}$, $\Pr\{ \mathrm{TP} \}$, and  $\Pr\{ \mathrm{AP} \}$ are\SJtodo{Need to rewrite}
%\begin{align*}
%\Pr\{ \mathrm{NP} \} &\in \left[ \Pr( Y = 0 \mid D=1), \Pr( Y = 0) \right], \\
%\Pr\{ \mathrm{AP} \} &\in \left[ \Pr( Y = 1 \mid D=0), \Pr( Y = 1) \right],
%\end{align*}
%while $\Pr\{ \mathrm{TP} \} = 1 - \Pr\{ \mathrm{NP} \} - \Pr\{ \mathrm{AP} \}$.

\subsubsection*{Identification of $\theta$ and $\rtheta$}
Note first that 
\begin{align} 
\theta &= \frac{(q_2+q_3+q_5+q_6) - (q_3+q_6)}{1 - (q_3+q_6)} = \frac{q_2+q_5}{1-P_{10}-P_{11}+q_5},   \label{eq:theta}\\
\rtheta &= \frac{(q_2+q_3+q_5+q_6) - (q_3+q_6)}{q_2+q_3+q_5+q_6} = \frac{q_2+q_5}{q_2+P_{10}+P_{11}}. \label{eq:rtheta}
\end{align} 
Therefore, we can obtain the upper bounds on $\theta$ and $\rtheta$ by maximizing \Cref{eq:theta,eq:rtheta}, respectively, subject to the constraints in \eqref{eq:mts0good} and \eqref{eq:mts1good}; the lower bounds are similar. Below we focus on the upper bound on $\theta$. 

%For the lower bound on $\theta$, consider 
%\[
%\min_{q_2,q_5} \ Q_\theta(q_2,q_5):= \frac{q_2+q_5}{1-P_{10}-P_{11} + q_5} \quad \text{subject to \eqref{eq:mts0good} and \eqref{eq:mts1good}}.
%\]
%Since $Q_\theta(\cdot,q_5)$ is monotonic, it suffices to solve 
%\[
%    \min_{q_5} \ Q_\theta(0,q_5) \quad \text{subject to \eqref{eq:mts0good}},
%\]
%from which it follows that the lower bound on $\theta$ is $Q_\theta(0,0) = 0$. Sharpness follows from continuity of $Q_\theta$. 

Consider 
\[
\max_{q_2,q_5} \ Q_\theta(q_2,q_5):= \frac{q_2+q_5}{1-P_{10}-P_{11} + q_5} \quad \text{subject to \eqref{eq:mts0good} and \eqref{eq:mts1good}},
\]
for which it suffices to consider %\SJtodo{fixed a typo in the displayed eq.}
\[
\max_{q_5} \ Q_\theta\left(P_{00} - \frac{P_{01}(P_{10}+P_{00})}{P_{01}+ P_{11}}, q_5 \right) \quad \text{subject to \eqref{eq:mts0good}}.
\]
Therefore, it follows that the upper bound on $\theta$ is given by 
\begin{equation}\label{eq:maxQtheta} 
Q_\theta\left( P_{00} - \frac{P_{01}(P_{10}+P_{00})}{P_{01}+ P_{11}},\ P_{11} - \frac{P_{10}(P_{01}+P_{11})}{P_{00}+P_{10}}\right).
\end{equation} 
Simplifying \Cref{eq:maxQtheta} yields $\theta_U$. Sharpness follows from continuity. \qed

%Let $Q_d:= P_{0d}+P_{1d}$, and simplifying this leads to 
%\begin{align*} 
%\frac{P_{00} - \frac{P_{01}Q_0}{Q_1} + P_{11} - \frac{P_{10}Q_1}{Q_0}}{1-P_{10} - %\frac{P_{10}Q_1}{Q_0}}
%&= 
%\frac{Q_0 P_{00} - \frac{P_{01}Q_0^2}{Q_1} + Q_0P_{11} - P_{10}Q_1}{Q_0 - P_{10}} \\
%&= 
%\frac{Q_0 P_{00} - \frac{P_{01}Q_0^2}{Q_1} + Q_0P_{11} + P_{10 }Q_0 - P_{10}}{Q_0 - %P_{10}} \\
%&= 
%\frac{\frac{1}{Q_1}\left( Q_0Q_1P_{00} - P_{01}Q_0^2 + Q_0Q_1P_{11}+Q_0Q_1P_{10} %\right) - P_{10}}{Q_0 - P_{10}} \\
%&= 
%\frac{\frac{1}{Q_1}\left\{ Q_0Q_1 (1-P_{01}) - P_{01}Q_0^2  \right\} - P_{10}}{Q_0 %- P_{10}} \\ 
%&= 
%\frac{\frac{Q_0}{Q_1}( Q_1 -P_{01} ) - P_{10}}{Q_0 - P_{10}} \\
%&= 
%\frac{\frac{Q_0}{Q_1}P_{11} - P_{10}}{Q_0 - P_{10}}  \\ 
%&= 
%\frac{P_{11}/Q_1 - P_{10}/Q_0}{1 - P_{10}/Q_0}.
%\end{align*} 
%Sharpness follows from continuity.   \qed

%Now, for the upper bound on $\rtheta$, we solve 
%\[
%\max_{q_2,q_5} \ Q_\rtheta(q_2,q_5):= \frac{q_2+q_5}{P_{10}+P_{11} + q_2} \quad %\text{subject to \eqref{eq:mts0good} and \eqref{eq:mts1good}}.
%\]
%By the same reasoning as before, the maximum value is given by 
%\[
%Q_\rtheta\left( P_{00} - \frac{P_{01}(P_{10}+P_{00})}{P_{01}+ P_{11}},\ P_{11} - %\frac{P_{10}(P_{01}+P_{11})}{P_{00}+P_{10}}\right)
%=
%\frac{P_{11}/Q_1 - P_{10}/Q_0}{1-P_{01}/Q_1}.
%\]
%The lower bound on $\rtheta$ is easy, and it will be omitted.   

\subsection*{Proof of Theorem \ref{thm:main:pc}: }
The case of $\mathrm{PNS}$ immediately follows from \Cref{thm:main} because $\mathrm{PNS} = \Pr\{\mathrm{TP}\}$. For $\mathrm{PS}$ and $\mathrm{PN}$, use the $q_j$'s defined in the proof of \Cref{thm:main} to obtain
%\begin{align*} 
%\mathrm{PS} =  \frac{\Pr\{ \mathrm{TP} , D=0 \}}{\Pr\{ \mathrm{TP} , D=0 \} + \Pr\{ %\mathrm{NP} , D=0 \}},\quad 
%\mathrm{PN} =  \frac{\Pr\{ \mathrm{TP} , D=1 \}}{\Pr\{ \mathrm{TP} , D=1 \} + \Pr\{ %\mathrm{AP} , D=1 \}}.
%\end{align*}
%we can express PS and PN by using the $q_j$'s defined in the proof of \Cref{thm:main}: i.e.,
\begin{align*} 
\mathrm{PS} =  \frac{q_2}{q_1 + q_2} =  \frac{q_2}{P_{00}}
\quad\text{and}\quad
\mathrm{PN} =  \frac{q_5}{q_5 + q_6} = \frac{q_5}{P_{11}}.
\end{align*}
Combine these with the constraints in \eqref{eq:mts0good} and \eqref{eq:mts1good}.
\qed

\end{appendix}
    
\newpage
    
    %%%%%%%%%%%%%%%%%%%%%%%%%%%%%%%%%%%%%%%%%%%%%%
    %% Bibliography:                            %%
    %%%%%%%%%%%%%%%%%%%%%%%%%%%%%%%%%%%%%%%%%%%%%%
    %% IMPORTANT: References in the bibliography should be complete, 
    %% including the first and last names, and date of publication.
    
    %% If your bibliography is in bibtex format, uncomment commands:
    \bibliographystyle{ecta-fullname} % Style BST file

    {\small \bibliography{persuasion3,casecontrol_SL}}  % Bibliography file (usually '*.bib')

\begin{thebibliography}{31}
\newcommand{\enquote}[1]{``#1''}
\expandafter\ifx\csname natexlab\endcsname\relax\def\natexlab#1{#1}\fi

\bibitem[\protect\citeauthoryear{Bhattacharya, Shaikh, and
  Vytlacil}{Bhattacharya et~al.}{2012}]{BSV-2012}
\textsc{Bhattacharya, Jay, Azeem~M. Shaikh, and Edward Vytlacil} (2012):
  \enquote{Treatment effect bounds: An application to {Swan-Ganz}
  catheterization,} \emph{Journal of Econometrics}, 168 (2), 223--243.

\bibitem[\protect\citeauthoryear{Canay and Shaikh}{Canay and
  Shaikh}{2017}]{canay2017practical}
\textsc{Canay, Ivan~A and Azeem~M Shaikh} (2017): \enquote{Practical and
  theoretical advances in inference for partially identified models,}
  \emph{Advances in Economics and Econometrics}, 2, 271--306.

\bibitem[\protect\citeauthoryear{Cavgias, Corbi, Meloni, and Novaes}{Cavgias
  et~al.}{2024}]{cavgias2024media}
\textsc{Cavgias, Alexsandros, Raphael Corbi, Luis Meloni, and Lucas~M. Novaes}
  (2024): \enquote{Media Manipulation in Young Democracies: Evidence From the
  1989 {Brazilian} Presidential Election,} \emph{Comparative Political
  Studies}, 57 (2), 221--253.

\bibitem[\protect\citeauthoryear{Dawid, Faigman, and Fienberg}{Dawid
  et~al.}{2014}]{Dawid2014fitting}
\textsc{Dawid, A.~Philip, David~L. Faigman, and Stephen~E. Fienberg} (2014):
  \enquote{Fitting Science Into Legal Contexts: Assessing Effects of Causes or
  Causes of Effects?} \emph{Sociological Methods \& Research}, 43 (3),
  359--390.

\bibitem[\protect\citeauthoryear{Dawid and Musio}{Dawid and
  Musio}{2022}]{Dawid2022effects}
\textsc{Dawid, A.~Philip and Monica Musio} (2022): \enquote{Effects of Causes
  and Causes of Effects,} \emph{Annual Review of Statistics and Its
  Application}, 9 (Volume 9, 2022), 261--287.

\bibitem[\protect\citeauthoryear{DellaVigna and Gentzkow}{DellaVigna and
  Gentzkow}{2010}]{dellavigna2010persuasion}
\textsc{DellaVigna, Stefano and Matthew Gentzkow} (2010): \enquote{Persuasion:
  empirical evidence,} \emph{Annual Review of Economics}, 2, 643--669.

\bibitem[\protect\citeauthoryear{DellaVigna and Kaplan}{DellaVigna and
  Kaplan}{2007}]{dellavigna2007fox}
\textsc{DellaVigna, Stefano and Ethan Kaplan} (2007): \enquote{The {Fox News}
  effect: Media bias and voting,} \emph{Quarterly Journal of Economics}, 122
  (3), 1187--1234.

\bibitem[\protect\citeauthoryear{Imbens and Manski}{Imbens and
  Manski}{2004}]{Imbens/Manski:04}
\textsc{Imbens, Guido and Charles~F. Manski} (2004): \enquote{Confidence
  Intervals for Partially Identified Parameters,} \emph{Econometrica}, 72 (6),
  1845--1857.

\bibitem[\protect\citeauthoryear{Ji, Lei, and Spector}{Ji
  et~al.}{2024}]{ji2023model}
\textsc{Ji, Wenlong, Lihua Lei, and Asher Spector} (2024):
  \enquote{Model-Agnostic Covariate-Assisted Inference on Partially Identified
  Causal Effects,} Ar{X}iv:2310.08115 [econ.EM], available at
  \url{https://arxiv.org/abs/2310.08115}.

\bibitem[\protect\citeauthoryear{Jun and Lee}{Jun and
  Lee}{2023}]{jun2023identifying}
\textsc{Jun, Sung~Jae and Sokbae Lee} (2023): \enquote{Identifying the effect
  of persuasion,} \emph{Journal of Political Economy}, 131 (8), 2032--2058.

\bibitem[\protect\citeauthoryear{Jun and Lee}{Jun and
  Lee}{2024{\natexlab{a}}}]{Jun2024JBES}
---\hspace{-.1pt}---\hspace{-.1pt}--- (2024{\natexlab{a}}): \enquote{Causal
  Inference Under Outcome-Based Sampling with Monotonicity Assumptions,}
  \emph{Journal of Business \& Economic Statistics}, 42 (3), 998--1009.

\bibitem[\protect\citeauthoryear{Jun and Lee}{Jun and
  Lee}{2024{\natexlab{b}}}]{jun2024aprt}
---\hspace{-.1pt}---\hspace{-.1pt}--- (2024{\natexlab{b}}): \enquote{Learning
  the Effect of Persuasion via Difference-In-Differences,} Ar{X}iv:2410.14871,
  available at \url{https://arxiv.org/abs/2410.14871}.

\bibitem[\protect\citeauthoryear{Kaji and Cao}{Kaji and
  Cao}{2025}]{kaji2023assessing}
\textsc{Kaji, Tetsuya and Jianfei Cao} (2025): \enquote{Assessing Heterogeneity
  of Treatment Effects,} Ar{X}iv:2306.15048 [econ.EM], available at
  \url{https://arxiv.org/abs/2306.15048}.

\bibitem[\protect\citeauthoryear{Kim, Kwon, Kwon, and Lee}{Kim
  et~al.}{2018}]{Kim-et-al:2018}
\textsc{Kim, Wooyoung, Koohyun Kwon, Soonwoo Kwon, and Sokbae Lee} (2018):
  \enquote{The identification power of smoothness assumptions in models with
  counterfactual outcomes,} \emph{Quantitative Economics}, 9 (2), 617--642.

\bibitem[\protect\citeauthoryear{Kreider and Pepper}{Kreider and
  Pepper}{2007}]{Kreider2007JASA}
\textsc{Kreider, Brent and John~V Pepper} (2007): \enquote{Disability and
  Employment,} \emph{Journal of the American Statistical Association}, 102
  (478), 432--441.

\bibitem[\protect\citeauthoryear{Kreider, Pepper, Gundersen, and
  Jolliffe}{Kreider et~al.}{2012}]{Kreider-et-al:2012}
\textsc{Kreider, Brent, John~V. Pepper, Craig Gundersen, and Dean Jolliffe}
  (2012): \enquote{Identifying the Effects of {SNAP} ({F}ood {S}tamps) on Child
  Health Outcomes When Participation Is Endogenous and Misreported,}
  \emph{Journal of the American Statistical Association}, 107 (499), 958--975.

\bibitem[\protect\citeauthoryear{Le~Pennec and Pons}{Le~Pennec and
  Pons}{2023}]{TVdebate:23}
\textsc{Le~Pennec, Caroline and Vincent Pons} (2023): \enquote{How do Campaigns
  Shape Vote Choice? Multicountry Evidence from 62 Elections and 56 TV
  Debates,} \emph{Quarterly Journal of Economics}, 138 (2), 703--767.

\bibitem[\protect\citeauthoryear{Machado, Shaikh, and Vytlacil}{Machado
  et~al.}{2019}]{MSV-2019}
\textsc{Machado, Cecilia, Azeem~M. Shaikh, and Edward~J. Vytlacil} (2019):
  \enquote{Instrumental variables and the sign of the average treatment
  effect,} \emph{Journal of Econometrics}, 212 (2), 522--555.

\bibitem[\protect\citeauthoryear{Manski}{Manski}{1997}]{manski1997}
\textsc{Manski, Charles~F.} (1997): \enquote{Monotone Treatment Response,}
  \emph{Econometrica}, 65 (6), 1311--1334.

\bibitem[\protect\citeauthoryear{Manski and Pepper}{Manski and
  Pepper}{2000}]{manski2000monotone}
\textsc{Manski, Charles~F and John~V Pepper} (2000): \enquote{Monotone
  instrumental variables: With an application to the returns to schooling,}
  \emph{Econometrica}, 68 (4), 997--1010.

\bibitem[\protect\citeauthoryear{Okumura and Usui}{Okumura and
  Usui}{2014}]{okumura2014concave}
\textsc{Okumura, Tsunao and Emiko Usui} (2014): \enquote{Concave-monotone
  treatment response and monotone treatment selection: With an application to
  the returns to schooling,} \emph{Quantitative Economics}, 5 (1), 175--194.

\bibitem[\protect\citeauthoryear{Pearl}{Pearl}{1999}]{pearl1999probabilities}
\textsc{Pearl, Judea} (1999): \enquote{Probabilities of causation: three
  counterfactual interpretations and their identification,} \emph{Synthese},
  121 (1-2), 93--149.

\bibitem[\protect\citeauthoryear{Pearl}{Pearl}{2009}]{pearl2009causality}
---\hspace{-.1pt}---\hspace{-.1pt}--- (2009): \emph{Causality}, Cambridge
  university press.

\bibitem[\protect\citeauthoryear{Pearl}{Pearl}{2015}]{Pearl2015causes}
---\hspace{-.1pt}---\hspace{-.1pt}--- (2015): \enquote{Causes of Effects and
  Effects of Causes,} \emph{Sociological Methods \& Research}, 44 (1),
  149--164.

\bibitem[\protect\citeauthoryear{Possebom and Riva}{Possebom and
  Riva}{2025}]{possebom2022probability}
\textsc{Possebom, Vitor and Flavio Riva} (2025): \enquote{Probability of
  Causation with Sample Selection: A Reanalysis of the Impacts of {Jóvenes en
  Acción} on Formality,} \emph{Journal of Business \& Economic Statistics},
  forthcoming, available at
  \url{https://doi.org/10.1080/07350015.2024.2388639}.

\bibitem[\protect\citeauthoryear{Stoye}{Stoye}{2009}]{Stoye:07}
\textsc{Stoye, J\"{o}rg} (2009): \enquote{More on Confidence Regions for
  Partially Identified Parameters,} \emph{Econometrica}, 77 (4), 1299--1315.

\bibitem[\protect\citeauthoryear{Tian and Pearl}{Tian and
  Pearl}{2000}]{tian2000probabilities}
\textsc{Tian, Jin and Judea Pearl} (2000): \enquote{Probabilities of causation:
  Bounds and identification,} \emph{Annals of Mathematics and Artificial
  Intelligence}, 28 (1), 287--313.

\bibitem[\protect\citeauthoryear{Vytlacil and Yildiz}{Vytlacil and
  Yildiz}{2007}]{vytlacil2007dummy}
\textsc{Vytlacil, Edward and Nese Yildiz} (2007): \enquote{Dummy endogenous
  variables in weakly separable models,} \emph{Econometrica}, 75 (3), 757--779.

\bibitem[\protect\citeauthoryear{Yamamoto}{Yamamoto}{2012}]{Yamamoto:2012}
\textsc{Yamamoto, Teppei} (2012): \enquote{Understanding the Past: Statistical
  Analysis of Causal Attribution,} \emph{American Journal of Political
  Science}, 56 (1), 237--256.

\bibitem[\protect\citeauthoryear{Yu}{Yu}{2025}]{yu2023binary}
\textsc{Yu, Zeyang} (2025): \enquote{A Binary IV Model for Persuasion:
  Profiling Persuasion Types among Compliers,} \emph{Econometrics Journal},
  forthcoming, available at \url{ https://doi.org/10.1093/ectj/utaf003}.

\bibitem[\protect\citeauthoryear{Zhang, Geng, Li, and Ding}{Zhang
  et~al.}{2024}]{Ding:2024:arXiv:prob_necessity}
\textsc{Zhang, Chao, Zhi Geng, Wei Li, and Peng Ding} (2024):
  \enquote{Identifying and bounding the probability of necessity for causes of
  effects with ordinal outcomes,} Ar{X}iv:2411.01234 [math.ST], available at
  \url{https://arxiv.org/abs/2411.01234}.

\end{thebibliography}

\newpage

\begin{figure}[htbp!]
    \centering
    
% --- PANEL A: THE TABLE ---
{\small Panel A: Examples of Vote Choice Consistency ($Y$)}
\par\medskip
\setlength{\tabcolsep}{4pt} % Slightly increased for readability
\begin{tabular}{l ccccc}
    \hline\hline
    Pre-Election Survey & \multicolumn{2}{c}{Survey Choice} & \multicolumn{2}{c}{Potential Outcomes} & Outcome \\
    \cline{2-3} \cline{4-5}
    Timing vs. Debate & Pre-Election & Post-Election & $Y(0)$ & $Y(1)$ & $Y$ \\
    \hline
    \multicolumn{6}{l}{\textit{Scenario 1: MTR holds (Clarification)}} \\
    Before ($D=0$) & Undecided & Cand. A & 0 & -- & 0 \\
    After \phantom{l}($D=1$)  & Cand. A & Cand. A & -- & 1 & 1 \\
    \hline
    \multicolumn{6}{l}{\textit{Scenario 2: MTR violation (Confusion)}} \\
    Before ($D=0$) & Cand. A & Cand. A & 1 & -- & 1 \\
    After \phantom{l}($D=1$)  & Undecided & Cand. A & -- & 0 & 0 \\
    \hline
    \multicolumn{6}{l}{\textit{Scenario 3: Dropout}} \\
    Before ($D=0$) & Cand. A & No Vote & 0 & -- & 0 \\
    After \phantom{l}($D=1$)  & Cand. A & No Vote & -- & 0 & 0 \\
    \hline
    \multicolumn{6}{l}{\textit{Scenario 4: Initially Undecided}} \\
    Before ($D=0$) & Undecided & Any & 0 & -- & 0 \\
    After \phantom{l}($D=1$)  & Undecided & Any & -- & 0 & 0 \\
    \hline\hline
\end{tabular}

\medskip
\begin{minipage}{0.9\linewidth} 
    \raggedright 
    \footnotesize 
    \emph{Notes:} Each row represents a single respondent surveyed twice. The first column indicates whether the pre-election interview occurred before or after the debate, which determines treatment status ($D$). $Y(0)$ is the potential outcome if surveyed before the debate; $Y(1)$ if after. $Y=1$ indicates consistent support for Candidate A (Pre-Election Intent matches Post-Election Vote).
\end{minipage}    
    
    \vspace{1cm} % Spacing between panels

    % --- PANEL B: THE FIGURE ---
    {\small Panel B: Definition of Treatment Variable ($D$) and Sample Window}
    \par\medskip
    \begin{tikzpicture}[
        scale=1.35, % Slightly reduced to fit width
        % Styles for consistent formatting
        axis/.style={->, >=Stealth, thick, black!70},
        node_style/.style={circle, draw, thick, minimum size=0.8cm, inner sep=0pt, font=\small},
        control_node/.style={node_style, draw=gray!80, fill=gray!10},
        treat_node/.style={node_style, draw=black, fill=black!20}
    ]

    % 1. Draw the Main Time Axis
    \draw[axis] (-4, 0) -- (4, 0) node[right, font=\footnotesize] {Day};

    % 2. Draw Control Group Nodes (-3 to -1)
    \foreach \x in {-3, -2, -1} {
        \node[control_node] (d\x) at (\x, 0) {\x};
    }

    % 3. Draw Treatment Group Nodes (0 to +3)
    \foreach \x in {0, 1, 2, 3} {
        \node[treat_node] (d\x) at (\x, 0) {+\x};
    }

    % Overwrite day 0 to remove "+" and emphasize it
    \node[treat_node, line width=1.5pt] (d0) at (0, 0) {0};

    % 4. Annotate Debate Day
    % Positioned above the node but clearly under the top "Window" brace
    \node[above=0.4cm of d0, font=\footnotesize] (label_debate) {Debate Day};
    \draw[dotted, thick] (d0.north) -- (label_debate.south);

    % 5. Braces and Group Labels
    
    % Control Brace (Bottom)
    \draw[decorate, decoration={brace, amplitude=10pt, mirror, raise=5pt}, thick, gray] 
        (d-3.south west) -- (d-1.south east) 
        node[midway, below=20pt, font=\small, text=gray] {Control Group ($D=0$)};

    % Treatment Brace (Bottom)
    \draw[decorate, decoration={brace, amplitude=10pt, mirror, raise=5pt}, thick, black] 
        (d0.south west) -- (d3.south east) 
        node[midway, below=20pt, font=\small] {Treatment Group ($D=1$)};

    % 6. Sample Window Label (Top)
    % Raised high enough (35pt) to encompass the "Debate Day" label below it
    \draw[decorate, decoration={brace, amplitude=10pt, raise=35pt}, thick, black!60] 
        (d-3.north west) -- (d3.north east) 
        node[midway, above=45pt, font=\small\scshape] {Restricted Sample Window $[-3, +3]$};

    \end{tikzpicture}

    \caption{Data Construction. Panel A illustrates the coding rules for the outcome variable $Y$. Panel B illustrates the sample restriction and treatment assignment relative to the debate timing.}
    \label{fig:combined_methodology}
\end{figure}

\newpage

\begin{table}[!htbp]
  \centering
  \caption{Sharp Bounds on $\theta$, $\theta^{(r)}$, PN, PS, and PNS} 
  \label{tab:bounds_summary}
  \bigskip

  \begin{tabular}{ccccccc}
    \hline\hline
    MTR & MTS & $\theta$ & $\theta^{(r)}$ & PN & PS & PNS \\
    \hline
    No  & No  & $[0,1]$ & $[0,1]$ & $[0,1]$ & $[0,1]$ & $[0, \overline{\textrm{PNS}}_{\text{NA}}]$ \\
    Yes & No  & $[0,1]$ & $[0,1]$ & $[0,1]$ & $[0,1]$ & $[0, \overline{\textrm{PNS}}_{\text{NA}}]$ \\
    No  & Yes & $[0,\overline{\theta}_{\text{MTS}}]$ & $[0,\overline{\theta}^{(r)}_{\text{MTS}}]$ & $[0,\overline{\textrm{PN}}_{\text{MTS}}]$ & $[0,\overline{\textrm{PS}}_{\text{MTS}}]$ & $[0,\overline{\textrm{PNS}}_{\text{MTS}}]$ \\
    Yes & Yes & $[0,\theta_U]$ & $[0,\theta^{(r)}_U]$ & $[0,\theta^{(r)}_U]$ & $[0,\theta_U]$ & $[0, \overline{\textrm{ATE}}]$ \\
    \hline
  \end{tabular}
  
  \medskip
  \begin{minipage}{0.9\textwidth}
    {\small \emph{Notes}: MTR and MTS refer to Monotone Treatment Response and Monotone Treatment Selection, respectively. Rows 1--3 correspond to the cases derived in \Cref{thm:alternative_bounds,thm:alternative_bounds_pearl_tian}. Row 4 corresponds to the main results in \Cref{thm:main,thm:main:pc}. The bounds are defined by:
    \begin{align*}
      \overline{\mathrm{PNS}}_{\text{NA}} &:= \Pr(Y=1,D=1) + \Pr(Y=0,D=0), \\
      \overline{\mathrm{PNS}}_{\text{MTS}} &:= \min\{ \Pr(Y=0\mid D=0),\ \Pr(Y=1\mid D=1) \}, \\
      \overline{\mathrm{PS}}_{\text{MTS}} &:= \min\left\{ \frac{\Pr(Y=1\mid D=1)}{\Pr(Y=0\mid D=0)},\ 1 \right\}, \\
      \overline{\mathrm{PN}}_{\text{MTS}} &:= \min\left\{ \frac{\Pr(Y=0\mid D=0)}{\Pr(Y=1\mid D=1)},\ 1 \right\}, \\
      \overline{\theta}_{\text{MTS}} &:= \min\left\{\frac{\Pr(Y=1\mid D=1)}{\Pr(Y=1,D=1)+\Pr(Y=0,D=0)},\ 1 \right\}, \\
      \overline{\theta}^{(r)}_{\text{MTS}} &:= \min\left\{\frac{\Pr(Y=0\mid D=0)}{\Pr(Y=1,D=1)+\Pr(Y=0,D=0)},\ 1 \right\}, \\
      \overline{\textrm{ATE}} &:= \Pr(Y=1\mid D=1) - \Pr(Y=1\mid D=0), \\ 
          \theta_U &:= \frac{\overline{\textrm{ATE}}}{\Pr(Y=0\mid D=0)}, \\
  \theta_U^{(r)} &:= \frac{\overline{\textrm{ATE}}}{\Pr(Y=1\mid D=1)}.
    \end{align*}
    }
  \end{minipage}
\end{table}

\newpage

\begin{table}[!htbp]
  \centering
  \caption{Estimation Results}
  \label{tab:combined_results}

\bigskip

  % --- Panel A ---
  Panel A: Upper Bounds on the Persuasion Rates (in Percentage)
  \medskip
  \begin{tabular}{l r rrr rrr rrr}
    \hline\hline
    & & \multicolumn{3}{c}{ATE} & \multicolumn{3}{c}{APR ($\theta$)} & \multicolumn{3}{c}{R-APR ($\rtheta$)} \\
    & $N$ & Est. & SE & UCB & Est. & SE & UCB & Est. & SE & UCB \\ \hline
    All & 42136 & 0.68 & 0.56 & 1.60 & 3.37 & 2.76 & 7.91 & 0.84 & 0.69 & 1.98 \\
    Age $<$ 50 & 18418 & 1.06 & 0.85 & 2.45 & 4.83 & 3.81 & 11.10 & 1.34 & 1.06 & 3.09 \\
    Age $\ge$ 50 & 23718 & 0.40 & 0.58 & 1.36 & 2.13 & 3.09 & 7.22 & 0.49 & 0.71 & 1.66 \\
    U.S. & 12002 & 0.65 & 0.58 & 1.60 & 5.09 & 4.45 & 12.42 & 0.74 & 0.65 & 1.81 \\
    U.K. & 19454 & 0.61 & 0.90 & 2.08 & 3.69 & 5.43 & 12.63 & 0.72 & 1.06 & 2.46 \\
    \hline
  \end{tabular}

  \bigskip

  % --- Panel B ---
  Panel B: Proportions of AP and NP (in Percentage)
  \medskip
  \begin{tabular}{lcccccccc}
    \hline\hline
    & \multicolumn{4}{c}{Already-Persuaded (AP)} & \multicolumn{4}{c}{Never-Persuadable (NP)} \\
    & CI-LB & LB & UB & CI-UB & CI-LB & LB & UB & CI-UB \\ \hline
    All & 76.93 & 79.88 & 80.25 & 83.19 & 16.41 & 19.44 & 19.75 & 22.71 \\
    Age $<$ 50 & 74.78 & 78.04 & 78.62 & 81.85 & 17.49 & 20.90 & 21.38 & 24.65 \\
    Age $\ge$ 50 & 78.50 & 81.30 & 81.52 & 84.29 & 15.45 & 18.30 & 18.48 & 21.26 \\
    U.S. & 86.06 & 87.30 & 87.65 & 88.83 & 10.72 & 12.06 & 12.35 & 13.55 \\
    U.K. & 81.09 & 83.59 & 83.93 & 86.48 & 12.95 & 15.80 & 16.07 & 18.66 \\
    \hline
  \end{tabular}

  \medskip
  \begin{minipage}{0.9\textwidth}
    {\small \emph{Notes}: Panel A presents the upper bound estimates, standard errors (SE), and 95\% upper confidence bounds (UCB) for ATE, APR, and R-APR. Panel B reports lower and upper bounds (LB and UB, respectively) on the proportions of Already-Persuaded (AP) and Never-Persuadable (NP) types along with 95\% confidence lower and upper bounds (CI-LB and CI-UB, respectively). All figures are in percentages. Standard errors are clustered at the debate level.}
  \end{minipage}
\end{table}

\newpage

\begin{figure}[htbp]
\centering
\usetikzlibrary{shapes.geometric, arrows.meta, positioning, calc}

\tikzset{
    % Start/Prerequisite Node (Pill Shape, Centered, No Bold)
    startnode/.style={
        rectangle, 
        rounded corners=15pt, 
        draw=black, 
        thick, 
        fill=gray!30, 
        text width=8cm, 
        minimum height=1cm, 
        text centered, 
        font=\small
    },
    % Decision Nodes (Left Column)
    decision/.style={
        rectangle, 
        rounded corners=3pt, 
        draw=black!80, 
        thick, 
        fill=gray!10, 
        text width=4.8cm, 
        minimum height=1.5cm, 
        text centered, 
        font=\small
    },
    % Outcome Nodes (Right Column)
    outcome/.style={
        rectangle, 
        draw=black, 
        very thick, 
        fill=white, 
        text width=5.5cm, 
        minimum height=1.5cm, 
        text centered, 
        font=\small
    },
    % Arrow Styles
    arrow/.style={
        thick, 
        ->, 
        >=stealth, 
        color=black!90
    },
    % Labels on lines
    labeltext/.style={
        font=\footnotesize\itshape, 
        fill=white, 
        inner sep=2pt
    }
}

\begin{tikzpicture}[node distance=1.0cm and 1.5cm]

% --- 1. DEFINE THE COLUMNS FIRST ---

% Left Column: First Decision
\node (d_sel) [decision] {Is sample selection a concern?\\ \textit{\footnotesize(i.e., are data missing for some units?)}};

% Left Column: Subsequent Decisions
\node (d_exo) [decision, below=of d_sel] {Is the assumption of exogenous treatment reasonable?};
\node (d_iv) [decision, below=of d_exo] {Are credible instrumental variables available?};
\node (d_panel) [decision, below=of d_iv] {Are panel data available?};
\node (d_mts) [decision, below=of d_panel] {Is the MTS assumption reasonable?};

% Right Column: Outcomes
\node (o_possebom) [outcome, right=of d_sel] {Use methods developed in\\
\citet{possebom2022probability}\\ \textit{\footnotesize(requires treatment exogeneity)}};
\node (o_direct) [outcome, right=of d_exo] {Use direct sample analogs\\ (see Lemma~\ref{lem:benchmark})};
\node (o_iv) [outcome, right=of d_iv] {Use methods developed in\\  \citet{jun2023identifying} and \citet{yu2023binary}};
\node (o_panel) [outcome, right=of d_panel] {Use methods developed in\\ \citet{jun2024aprt}\\ \textit{\footnotesize(requires parallel trends and no anticipation)}};
\node (o_bounds) [outcome, right=of d_mts] {Use the bounding method developed in this paper};

% --- 2. DEFINE THE START NODE (CENTERED ABOVE) ---

% Positioned above the midpoint of the top row
% Using calc library syntax: $(NodeA)!0.5!(NodeB)$ finds the midpoint
\node (start) [startnode, above=0.8cm of $(d_sel.north)!0.5!(o_possebom.north)$] {Prerequisites:\\ Binary Outcome, Binary Treatment, Covariates\\ \& MTR Assumption};

% --- 3. ARROWS & PATHS ---

% Path: Selection -> Yes -> Possebom
\draw [arrow] (d_sel) -- node[labeltext, midway, above] {Yes} (o_possebom);

% Path: Selection -> No -> Exogeneity
\draw [arrow] (d_sel) -- node[labeltext, midway, right] {No} (d_exo);

% Path: Exogeneity -> Yes -> Direct Analogs
\draw [arrow] (d_exo) -- node[labeltext, midway, above] {Yes} (o_direct);

% Path: Exogeneity -> No -> IV
\draw [arrow] (d_exo) -- node[labeltext, midway, right] {No} (d_iv);

% Path: IV -> Yes -> Jun/Lee 2023
\draw [arrow] (d_iv) -- node[labeltext, midway, above] {Yes} (o_iv);

% Path: IV -> No -> Panel
\draw [arrow] (d_iv) -- node[labeltext, midway, right] {No} (d_panel);

% Path: Panel -> Yes -> Jun/Lee 2024b
\draw [arrow] (d_panel) -- node[labeltext, midway, above] {Yes} (o_panel);

% Path: Panel -> No -> MTS
\draw [arrow] (d_panel) -- node[labeltext, midway, right] {No} (d_mts);

% Path: MTS -> Yes -> This Paper
\draw [arrow] (d_mts) -- node[labeltext, midway, above] {Yes} (o_bounds);

\end{tikzpicture}

\vspace{0.2cm}
\caption{Decision Rule for Estimating Persuasion Rates}
\label{fig:decision_tree}

\vspace{0.2cm}
\begin{minipage}{0.85\textwidth}
\footnotesize
\textit{Notes:} This figure illustrates the recommended hierarchy of econometric methods. All methods assume a setting with binary outcomes, binary treatments, and observed covariates under the Monotone Treatment Response (MTR) assumption. MTS denotes Monotone Treatment Selection.
\end{minipage}

\end{figure}

%%%%%%%%%%%%%%%%%%%%%%%%%%%%%%%%%%%%%%%%%%%%%%%%%%%%
%%%%%%%  Online Appendix if needed
%%%%%%%%%%%%%%%%%%%%%%%%%%%%%%%%%%%%%%%%%%%%%%%%%%%%%
\begin{appendix}
        \clearpage

        \pagenumbering{roman}
        \setcounter{page}{1}
        
        \renewcommand{\thesection}{S-\arabic{section}}
        \setcounter{section}{0}

        \begin{center}
            \Large{Online Appendices to ``Bounding the Effect of Persuasion with Monotonicity Assumptions: Reassessing the Impact of TV Debates''}
        \end{center}
        \bigskip
        \begin{center} 
            \begin{tabular}{ccc}
                Sung Jae Jun & \qquad & Sokbae Lee  \\
                Penn State University & \qquad  & Columbia University
                \end{tabular}
        \end{center}        
        \bigskip

\section{The Empirical Framework of \citet{TVdebate:23}}\label{appnx:LP:reg}

In the analysis by \citet{TVdebate:23}, the unit of observation is a respondent-debate-election triad, yielding a total of 331,000 observations. Their main regression specification is defined (in their notation) as follows:
\begin{align}\label{eq:reg:tv-debate}
Y_{i t}^d = \sum_{k=-3}^{-1} \mu_k + \sum_{k=1}^3 \mu_k + \mu_{4^{-}} + \mu_{4^{+}} + \sum_{t=-60}^{-1} \beta_t D_t + \alpha^d + W_{i t}^{\prime} \lambda + u_{i t}^d,
\end{align}
where $Y_{i t}^d$ is the outcome for respondent $i$, surveyed $t$ days before the election associated with debate $d$. The term $D_t$ (distinct from the treatment variable in our notation) represents a set of indicator variables controlling for the number of days prior to the election. Here, $W_{i t}$ is a vector of covariates including day-of-the-week effects, post-electoral survey lag fixed effects, and sociodemographic characteristics.

The coefficients $\mu_k$ (for $-3 \leq k \leq 3$) correspond to the effects of indicator variables denoting the number of days relative to the debate. Additionally, $\mu_{4^{-}}$ and $\mu_{4^{+}}$ capture the effects of being surveyed four or more days before or after the debate, respectively, while $\alpha^d$ denotes debate-by-election fixed effects.

The primary coefficients of interest are $\mu_1$, $\mu_2$, and $\mu_3$, which capture the effect of the debate one to three days post-event relative to the omitted category, $\mu_0$. \citet{TVdebate:23} use the day of the debate as the reference category because debates typically occur in the evening. This specification assumes that any pre-trends prior to four days before the debate, or effects persisting beyond four days after, are adequately absorbed by the fixed effects $\mu_{4^{-}}$ and $\mu_{4^{+}}$.

Identification relies on the validity of the linear regression model in \Cref{eq:reg:tv-debate}; specifically, it requires that the timing of the debate be uncorrelated with the regression errors, conditional on the controls. Standard errors are clustered at the debate level.

\section{Proof of \Cref{thm:alternative_bounds}}

\noindent
\textit{Proof of cases (i) and (ii).}
	First, for cases (i) and (ii), it suffices to prove the results in case (ii). Therefore, suppose that MTR holds but MTS may not.   In this case, we have 
	\begin{align} 
		\theta &= \frac{\Pr\{Y(1) = 1\} - \Pr\{ Y(0) = 1\}}{1-\Pr\{Y(0) = 1 \}}, \label{MTR1}\\
		\theta^{(r)} &= \frac{\Pr\{Y(1) = 1\} - \Pr\{ Y(0) = 1\}}{\Pr\{Y(1) = 1 \}}. \label{MTR2}
	\end{align} 
	Here, all we know about the marginals will be 
	\begin{align}
	\Pr(Y=1)
	&= 
	\Pr( Y=1, D=1) + \Pr\{ Y(0) = 1, D=0\} \notag \\
	&\leq
	\Pr( Y=1, D=1) + \Pr\{ Y(1) = 1, D=0\}
	= 
	\Pr\{ Y(1) = 1\} 
	\leq 1,  \label{onlyMTR1}
	\end{align}
	and 
	\begin{align} 
	0
	\leq \Pr\{ Y(0) = 1\} 
	&= \Pr(Y=1, D=0) + \Pr\{ Y(0) = 1, D=1 \} \notag \\
	&\leq \Pr(Y=1, D=0) + \Pr\{ Y = 1, D=1 \} = \Pr(Y=1). \label{onlyMTR2}	
	\end{align} 
	The inequalities in \Cref{onlyMTR1,onlyMTR2} are sharp  under MTR only.
	Therefore, the sharp upper bound on $\theta$ with MTR but without MTS can be obtained by solving 
	\[
	\max_{a,b}\ \frac{a-b}{1-b}\ \text{  subject to  }\ \Pr(Y=1)\leq a\leq 1,\ 0\leq b\leq \Pr(Y=1).
	\]
	The solution is $1$. Similarly, the sharp lower bound on $\theta$ under MTR but no MTS is a solution to 
	\[
	\min_{a,b}\ \frac{a-b}{1-b}\ \text{  subject to  }\ \Pr(Y=1)\leq a\leq 1,\ 0\leq b\leq \Pr(Y=1),
	\]
	which is simply $0$. 
Therefore, the sharp bounds on $\theta$ without MTS (with or without MTR) will be just trivial. The case of $\theta^{(r)}$ is similar: the upper bound is $1$ by solving
	\[
	\max_{a,b}\ \frac{a-b}{a}\ \text{  subject to  }\ \Pr(Y=1)\leq a\leq 1,\ 0\leq b\leq \Pr(Y=1),
	\]
	and the lower bound is $0$ by solving 
	\[
	\min_{a,b}\ \frac{a-b}{a}\ \text{  subject to  }\ \Pr(Y=1)\leq a\leq 1,\ 0\leq b\leq \Pr(Y=1).  \qedhere
	\]

\noindent
\textit{Proof of case (iii).}
    Now, we consider case (iii). Tabulate the PMF of $\bigl(Y(0), Y(1), D\bigr)$ without imposing MTR as follows:  
	\begin{center}
		\begin{tabular}{c|cccc}
			$\bigl( Y(0),Y(1) \bigr)$  &	 $(0,0)$	&	$(0,1)$	&	$(1,0)$	&	$(1,1)$  \\ \hline  
			$D=0$ & 	$q_1$  & $q_2$ 	&	$q_3$	&	$q_4$  \\ 
			$D=1$ &  	$q_5$ & $q_6$ 	&	$q_7$	&	$q_8$		\\ \hline
		\end{tabular} 
	\end{center}
    The MTS assumption can then be expressed by the following inequalities: 
	\begin{align} \label{MTS}
		\frac{q_6+q_8}{q_5+q_6+q_7+q_8} \geq \frac{q_2+q_4}{q_1+q_2+q_3+q_4},
		\quad 
		\frac{q_7+q_8}{q_5+q_6+q_7+q_8} \geq \frac{q_3+q_4}{q_1+q_2+q_3+q_4}
	\end{align} 
	Following the proof of Theorem 1, let $P_{yd} = \Pr(Y=y, D=d) \in (0,1)$ for $(y,d) \in \{0,1\}^2$, and all we can learn from the distribution of $(Y,D)$ can be summarized by the following equations: 
	\begin{align*} 
		&q_1 + q_2 = P_{00}, 
		\quad 
		q_3 + q_4 = P_{10}, 
		\\
		&q_5 + q_7 = P_{01},
		\quad 
		q_6 + q_8 = 1-P_{00}-P_{10}-P_{01}.
	\end{align*} 
	where $\sum_{j=1}^8 q_j = 1$.  Therefore, four of the $q_j$'s are undetermined here: we have 7 free variables and three non-redundant equations. Specifically, treating $q_2, q_4, q_7,q_8$ as undetermined variables, we can write
	\begin{equation}\label{qjs}
	q_1 = P_{00} - q_2,\ 
	q_3 = P_{10} - q_4,\  
	q_5 = P_{01} - q_7,\ 
	q_6 = P_{11} - q_8,
	\end{equation}
	where $0\leq q_2\leq P_{00}, 0\leq q_4 \leq P_{10}, 0\leq q_7\leq P_{01}$, and $0\leq q_8\leq P_{11}$. Then, the inequalities in \Cref{MTS} become 
	\begin{align} 
		\frac{P_{11}}{P_{01}+P_{11}} \geq \frac{q_2 + q_4}{P_{00}+P_{10}}, \quad 
		\frac{q_7+q_8}{P_{01}+P_{11}} \geq \frac{P_{10}}{P_{00}+P_{10}}.
	\end{align} 
	Therefore, all the information we have from MTS can be summarized by \Cref{qjs} and  
	\begin{equation}\label{constraints from MTS only}
    \begin{aligned}     
    &0\leq q_2\leq P_{00},\ 0\leq q_4 \leq P_{10},\  0\leq q_7\leq P_{01},\ 0\leq q_8\leq P_{11}, \\     
    &0 \leq q_2 + q_4 \leq \frac{P_{11}Q_0}{Q_1},\ 
	\frac{P_{10}Q_1}{Q_0} \leq q_7 + q_8 \leq Q_1,
    \end{aligned}
	\end{equation}
	where $Q_d = P_{0d} + P_{1d} = \Pr(D=d)$ for $d \in \{0,1\}$.

%    Note that 
%    \[
%    \theta = \frac{q_2 + q_6}{q_1+q_2+q_5+q_6} %= \frac{q_2 - q_8 + P_{11}}{P_{00}+P_{01}+P_{11}-q_7-q_8}
%    \quad \text{and}\quad 
%    \theta^{(r)} = \frac{q_2 + q_6}{q_2+q_4+q_6+q_8} %= \frac{q_2 - q_8 + P_{11}}{P_{11} + q_2 + q_4}. 
%    \]
    We start with $\theta$ first.  For the sharp lower bound, we look for the minimum value that
    \[
    \theta = \frac{q_2+q_6}{q_1 + q_2 + q_5+q_6} = \frac{q_2 - q_8 + P_{11}}{P_{00}+P_{01}+P_{11}-q_7-q_8}
    \]
    can take subject to the constraints in \eqref{constraints from MTS only}. However, this objective function will be equal to zero if $q_2 = 0, q_8 = P_{11}$, which is clearly feasible: any $q_4$ and $q_7$ that satisfy $0\leq q_4 \leq \min(P_{10}, P_{11}Q_0/Q_1)$ and $\max(0, P_{10}Q_1/Q_0 - P_{11})\leq q_7\leq P_{01}$  will be compatible. Therefore, the sharp lower bound on $\theta$ is equal to zero. 
        For the sharp upper bound, we consider
    \[
    \max\ \frac{q_2+q_6}{q_1 + q_2 + q_5+q_6} = \frac{q_2 - q_8 + P_{11}}{P_{00}+P_{01}+P_{11}-q_7-q_8}
    \]
    subject to the constraints in \eqref{constraints from MTS only}. Noting that $q_2$ and $(q_7,q_8)$ are not cross-constrained, we maximize this with respect to $q_2$ first. Since $q_2$ is constrained by 
    \[
    0\leq q_2\leq P_{00},\ 0\leq q_4\leq P_{10},\ 0\leq q_2+q_4\leq P_{11}Q_0/Q_1,
    \]
    the largest value $q_2$ can take is $\min(P_{00},\ P_{11}Q_0/Q_1)$ when $q_4 = 0$. Therefore, it suffices to maximize 
    \[
    Q_U(q_7,q_8)
	:=
	\frac{P_{11} + \min(P_{00},\ P_{11}Q_0/Q_1)-q_8}{P_{00}+P_{01}+P_{11}-q_7-q_8}
    =
    \frac{\min(P_{11}+P_{00},\ P_{11}/Q_1)-q_8}{P_{00}+P_{01}+P_{11}-q_7-q_8}
    \]
    subject to 
    \[
    0\leq q_7\leq P_{01},\ 0\leq q_8\leq P_{11},\ P_{10}Q_1/Q_0\leq q_7+q_8\leq Q_1.
    \]
    For any fixed value of $q_8$ in the feasible set, $q_7$ must satisfy 
    \[
    \max(0, P_{10}Q_1/Q_0 - q_8) \leq q_7\leq \min(P_{01},Q_1 - q_8) = P_{01},  
    \]
   where the last equality follows from the constraint $0\leq q_8\leq P_{11}$.
    Then, maximizing $Q_U(q_7, q_8)$ with respect to $q_7$ with fixing $q_8$ leads to
    \[
	Q_U(P_{01}, q_8)
	=
    \frac{\min(P_{11}+P_{00}, P_{11}/Q_1)-q_8}{P_{00}+P_{11}-q_8}.
    \]
Therefore, it suffices to solve 
\[
\max\ 	Q_U(P_{01},q_8)= \frac{\min(P_{11}+P_{00}, P_{11}/Q_1)-q_8}{P_{00}+P_{11}-q_8}
\]
subject to 
\[
0\leq q_8\leq P_{11},\ P_{10}Q_1/Q_0 - P_{01}\leq q_8\leq P_{11}.
\]
Since $Q_U(P_{01},q_8)$ is decreasing in $q_8$, the maximum is attained by $q_8 = \max(0, P_{10}Q_1/Q_0 - P_{01})$ with the maximum value 
\[
Q_U\left( P_{01},\max(0, P_{10}Q_1/Q_0 - P_{01})  \right).
\]
Below we show that this is simplified to the expression we have in the theorem statement. 

First, note that
\[
\Pr (Y = 0 \mid D = 0) \leq \Pr (Y = 1 \mid D = 1) 
\Longleftrightarrow 
\Pr (Y = 0 \mid D = 1) \leq \Pr (Y = 1 \mid D = 0), 
\]
which is equivalent to 
\[
P_{00} \leq P_{11}Q_0/Q_1
\Longleftrightarrow 
0 \leq P_{10}Q_1/Q_0 - P_{01}.
\]
Since 
\[
\min(P_{11}+P_{00}, P_{11}/Q_1)
=
P_{11}+ \min(P_{00}, P_{11}Q_0/Q_1), 
\]
we have that
\begin{align*}
&Q_U\left( P_{01},\max(0, P_{10}Q_1/Q_0 - P_{01})  \right) \\
&=
\frac{P_{11}+ \min(P_{00}, P_{11}Q_0/Q_1)-\max(0, P_{10}Q_1/Q_0 - P_{01})}{P_{00}+P_{11}-\max(0, P_{10}Q_1/Q_0 - P_{01})} \\
&=
\begin{cases}
\frac{P_{11}+ P_{00} - P_{10}Q_1/Q_0 - P_{01}}{P_{00}+P_{11}- P_{10}Q_1/Q_0 - P_{01}} = 1  & \text{if  $P_{00} \leq P_{11}Q_0/Q_1$} \\
\frac{P_{11}+ P_{11}Q_0/Q_1}{P_{00}+P_{11}} =  \frac{P_{11}/Q_1}{P_{11}+P_{00}}  &\text{if  $P_{00} > P_{11}Q_0/Q_1$} \\
\end{cases}
\\
&= \one\{ P_{00} \leq P_{11}Q_0/Q_1 \} + \frac{P_{11}/Q_1}{P_{11}+P_{00}}  \one\{ P_{00} > P_{11}Q_0/Q_1 \} \\
&= \one \left\{ 1 \leq \frac{P_{11}/Q_1}{P_{11}+P_{00}}  \right\} + \frac{P_{11}/Q_1}{P_{11}+P_{00}}  \one \left\{ 1 > \frac{P_{11}/Q_1}{P_{11}+P_{00}}  \right\} \\
&= \min \left\{ \frac{P_{11}/Q_1}{P_{11}+P_{00}},\ 1 \right\},
\end{align*}
where the second last equality follows the equivalence that 
\begin{align*} 
P_{00} \leq P_{11}Q_0/Q_1
&\Leftrightarrow
P_{00}Q_1 \leq P_{11} Q_0
\Leftrightarrow
P_{00}Q_1 \leq P_{11} (1-Q_1) 
\\
&\Leftrightarrow
\{ P_{00} + P_{11} \} Q_1 \leq P_{11}
\Leftrightarrow
1 \leq \frac{P_{11}/Q_1}{P_{00} + P_{11}},
\end{align*}
and the last equality is due to the fact that $a\cdot 1(a\leq b) + b\cdot 1(a>b) = \min(a,b)$ for $a, b \in \mathbb{R}$.
In sum, the sharp bounds on $\theta$ under MTS are given by 
\[
0\leq \theta \leq \min\left\{ \frac{\Pr(Y=1\mid D=1)}{\Pr(Y=1,D=1)+\Pr(Y=0,D=0)},\ 1  \right\}.
\]

Now, consider $\theta^{(r)}$.  For the sharp lower bound, we consider minimizing 
\[
\theta^{(r)} = \frac{q_2 + q_6}{q_2+q_4+q_6+q_8} = \frac{q_2 - q_8 +P_{11}}{q_2 + q_4+P_{11}}
\]
subject to the constraints in \eqref{constraints from MTS only}. However, this objective function is zero if $q_2 = 0$ and $q_8 = P_{11}$, which is clearly feasible as in the case of the lower bound on $\theta$. Therefore, the sharp lower bound on $\theta^{(r)}$ is zero. For the upper bound on $\theta^{(r)}$, we consider 
\[
\max\  \frac{q_2 - q_8 +P_{11}}{q_2 + q_4+P_{11}}
\]
subject to the constraints in \eqref{constraints from MTS only}. Noting that $q_8$ and $(q_2,q_4)$ are not cross-constrained, we maximize this with respect to $q_8$ first. Since $q_8$ is constrained by 
\[
0\leq q_7\leq P_{01},\ 0\leq q_8\leq P_{11}, \frac{P_{10}Q_1}{Q_0}\leq q_7+q_8\leq Q_1,
\]
the smallest value $q_8$ can take is $\max(0,P_{10}Q_1/Q_0 - P_{01} )$ when $q_7 = P_{01}$. Therefore, it suffices to maximize
\[
Q_U^{(r)}(q_2,q_4)
:=
\frac{q_2 +P_{11} - \max(0,\ P_{10}Q_1/Q_0 - P_{01} )}{q_2 + q_4+P_{11}}
%=
%\frac{\min(P_{11}, P_{11} + P_{01}- P_{10}Q_1/Q_0) + q_2}{P_{11}+q_2+q_4}
\]
subject to 
\[
0\leq q_2\leq P_{00},\ 0\leq q_4\leq P_{10},\ 0\leq q_2+q_4\leq P_{11}Q_0/Q_1.
\]
For any $q_2$ in the feasible set, $q_4$ must satisfy 
\[
0\leq q_4\leq \min(P_{10},\ P_{11}Q_0/Q_1-q_2).
\]
Since $Q_U^{(r)}(q_2,q_4)$ is decreasing in $q_4$ for any fixed $q_2$, it suffices to maximize 
\[
	Q_U^{(r)}(q_2,0)
	= 
	\frac{P_{11} - \max(0, P_{10}Q_1/Q_0 - P_{01}) + q_2}{P_{11}+q_2}
\] 
subject to 
\[
0\leq q_2\leq P_{00},\ 0\leq q_2 \leq P_{11}Q_0/Q_1.
\]
Since $Q_U^{(r)}(q_2,0)$ is increasing in $q_2$, the maximum is 
\[
	Q_U^{(r)}\left( \min(P_{00}, P_{11}Q_0/Q_1),0 \right).
\]
Note here that 
\[
P_{00} \leq P_{11}Q_0/Q_1
\Longleftrightarrow 
P_{10}Q_1/Q_0 - P_{01} \geq 0.
\]
Therefore, $Q_U^{(r)}\left( \min(P_{00}, P_{11}Q_0/Q_1),0 \right)$ is simplified to 
\[
Q_U^{(r)}\left( \min(P_{00}, P_{11}Q_0/Q_1),0 \right)
=
\min\left\{ \frac{1-P_{10}/Q_0}{P_{11}+P_{00}},\ 1 \right\}. 
\]
In sum, the sharp bounds on $\theta^{(r)}$ under MTS are given by 
\[
0\leq \theta^{(r)}\leq \min\left\{\frac{\Pr(Y=0\mid D=0)}{\Pr(Y=1,D=1)+\Pr(Y=0,D=0)},\ 1  \right\}.   \qedhere
\]

\section{Proof of \Cref{thm:alternative_bounds_pearl_tian}}

As mentioned in the main text, cases (i) and (ii) are given in \citet[Sections 4.2.1 and 4.4.1]{tian2000probabilities}, respectively. We provide the proof for case (iii) here.

Tabulate the general PMF of $\bigl(Y(0), Y(1), D\bigr)$ as follows:  

\bigskip
\begin{center}
		\begin{tabular}{c|cccc}
			$\bigl( Y(0),Y(1) \bigr)$  &	 $(0,0)$	&	$(0,1)$	&	$(1,0)$	&	$(1,1)$  \\ \hline  
			$D=0$ & 	$q_1$  & $q_2$ 	&	$q_3$	&	$q_4$  \\ 
			$D=1$ &  	$q_5$ & $q_6$ 	&	$q_7$	&	$q_8$		\\ \hline
		\end{tabular} 
\end{center}
\bigskip 

The MTS assumption can then be expressed by the following inequalities: 
	\begin{align} \label{MTS:new}
		\frac{q_6+q_8}{q_5+q_6+q_7+q_8} \geq \frac{q_2+q_4}{q_1+q_2+q_3+q_4},
		\quad 
		\frac{q_7+q_8}{q_5+q_6+q_7+q_8} \geq \frac{q_3+q_4}{q_1+q_2+q_3+q_4}
	\end{align} 
Following the proof of Theorem 1, let $P_{yd} = \Pr(Y=y, D=d) \in (0,1)$ for $(y,d) \in \{0,1\}^2$ and $Q_d = P_{1d} + P_{0d}$. Now, all we can learn from the distribution of $(Y,D)$ can be summarized by the following equations: 
	\begin{align*} 
		&q_1 + q_2 = P_{00}, 
		\quad 
		q_3 + q_4 = P_{10}, 
		\\
		&q_5 + q_7 = P_{01},
		\quad 
		q_6 + q_8 = 1-P_{00}-P_{10}-P_{01}.
	\end{align*} 
where $\sum_{j=1}^8 q_j = 1$.  Therefore, four of the $q_j$'s are undetermined here: we have 7 free variables and three non-redundant equations. Specifically, treating $q_2, q_4, q_7,q_8$ as undetermined variables, all the constraints we have under MTS are given by 
\begin{equation}\label{eq:allconstraints}
\begin{aligned} 
&0\leq q_2 \leq P_{00}, \quad
0\leq q_4\leq P_{10},\quad 
0\leq q_7\leq P_{01},\quad 
0\leq q_8\leq P_{11}, \\
&0\leq q_2 + q_4 \leq \frac{P_{11}Q_0}{Q_1},\quad 
\frac{P_{10}Q_1}{Q_0} \leq q_7+q_8\leq Q_1.
\end{aligned} 
\end{equation}
Now, recall that 
\begin{align*} 
\textrm{PS} = \frac{q_2}{P_{00}},\quad 
\textrm{PN} = \frac{q_6}{P_{11}} = \frac{P_{11}-q_8}{P_{11}},\quad 
\textrm{PNS} = q_2 + q_6 = q_2 - q_8 + P_{11}.
\end{align*} 

Consider PNS first. For the sharp lower bound, we minimize $q_2 + q_6 = P_{11}+q_2 - q_8 \geq 0$ subject to the constraints in \eqref{eq:allconstraints}. However, if 
\[
q_2 = 0, q_4 = 0, q_7 = \max\left\{ 0,\ \frac{P_{10}Q_1}{Q_0} - P_{11}\right\}, q_8 = P_{11}, 
\]
then we have $q_2 - q_8 + P_{11} = 0$. These $q$ values are in the feasible set, because 
\[
   \max\left\{ 0,\ \frac{P_{10}Q_1}{Q_0} - P_{11}\right\} \leq P_{01},\ 
   q_7 + q_8 = \max\{ P_{11},\ \frac{P_{10}Q_1}{Q_0}\} \leq Q_1.
\] 
Therefore, the sharp lower bound on PNS is zero.

For the sharp upper bound on PNS, we need to maximize $q_2+q_6 = P_{11}+q_2-q_8$ subject to the constraints in \eqref{eq:allconstraints}. For any fixed value of $q_8$ in the feasible set, the objective function $P_{11}+q_2-q_8$ is increasing in $q_2$. Therefore, it suffices to consider the largest value $q_2$ can take, which is $q_2 = \min\{ P_{00}, P_{11}Q_0/Q_1\}$ when $q_4 = 0$. Therefore, it suffices to maximize 
\[
P_{11} + \min\{ P_{00},\ P_{11}Q_0/Q_1\} - q_8
\] 
subject to 
\[
0\leq q_7\leq P_{01},\ 0\leq q_8\leq P_{11},\ P_{10}Q_1/Q_0 \leq q_7+q_8\leq Q_1.
\]
This objective function is decreasing in $q_8$, and hence the maximum is attained when $q_8$ takes the smallest value in the feasible set. That is, if $q_8 = \max\{0,\ P_{10}Q_1/Q_0 - P_{01}\}$ and $q_7 = P_{01}$, then we obtain the maximum
\[
\textrm{PNS}^{max}:=P_{11} + \min\{ P_{00},\ P_{11}Q_0/Q_1\} - \max\{0,\ P_{10}Q_1/Q_0 - P_{01}\}
\]
This expression can be simplified because 
\[
P_{00} \leq P_{11}Q_0/Q_1
\Longleftrightarrow
P_{00}P_{01} \leq P_{11}P_{10}
\Longleftrightarrow 
P_{10}Q_1/Q_0 \geq P_{01}.
\]
Therefore, 
\begin{align*}
\textrm{PNS}^{max}
&=
P_{11} + \min\{ P_{00},\ P_{11}Q_0/Q_1\} + \min\{ 0,\ P_{01}-P_{10}Q_1/Q_0 \}
\\
&=
P_{11} + \min\{P_{00}+P_{01}-P_{10}Q_1/Q_0,\ P_{11}Q_0/Q_1  \}
\\
&= 
\min\{P_{11}+P_{00}+P_{01}-P_{10}Q_1/Q_0,\ P_{11} + P_{11}Q_0/Q_1 \}
\\
&=
\min\{1-P_{10}-P_{10}Q_1/Q_0,\ P_{11} + P_{11}Q_0/Q_1 \}
\\
&=
\min\{1-P_{10}/Q_0,\ P_{11}/Q_1 \}.
\end{align*}

In sum, the sharp bounds on PNS under MTS only are given by 
\[
0\leq \textrm{PNS}\leq \min\{\Pr(Y=0\mid D=0), \Pr(Y=1\mid D=1) \}. 
\]

Now, consider PS. For the sharp upper bound on PS, it suffices to solve
\[
\max\ \frac{q_2}{P_{00}} \quad \text{subject to}\quad 0\leq q_2\leq P_{00},\ 0\leq q_4\leq P_{10},\ 0\leq q_2+q_4\leq \frac{P_{11}Q_0}{Q_1},
\]
because the other constraints are irrelevant. The maximum is 
\[
\frac{\min\left\{ P_{00},\ \frac{P_{11}Q_0}{Q_1} \right\}}{P_{00}} = \min\left\{ \frac{P_{11}Q_0}{P_{00}Q_1},\ 1 \right\},
\]
which is attained when 
\[
q_2 = \min\left\{P_{00}, \frac{P_{11}Q_0}{Q_1}\right\},\ 
q_4 = 0.
\]
For the sharp lower bound on PS, we consider 
\[
    \min\ \frac{q_2}{P_{00}} \quad \text{subject to}\quad 0\leq q_2\leq P_{00},\ 0\leq q_4\leq P_{10},\ 0\leq q_2+q_4\leq \frac{P_{11}Q_0}{Q_1},
\]
but this minimum is clearly zero because $q_2 = q_4 = 0$ is in the feasible set. Therefore, we conclude that the sharp bounds on PS under MTS (but without MTR) are given by 
\[
0\leq \textrm{PS}\leq \min\left\{ \frac{\Pr(Y=1\mid D=1)}{\Pr(Y=0\mid D=0)},\ 1 \right\}.
\]

Finally, consider PN. For the sharp upper bound on PN, it suffices to solve 
\[
\max\ \frac{P_{11} - q_8}{P_{11}} \quad \text{subject to}\quad 0\leq q_7\leq P_{01}, 0\leq q_8\leq P_{11}, \frac{P_{10}Q_1}{Q_0}\leq q_7+q_8\leq Q_1,
\]
because the other constraints are irrelevant. The maximum is 
\[
\frac{P_{11} - \max\left\{  0,\ \frac{P_{10}Q_1}{Q_0} - P_{01}\right\}}{P_{11}}
=
\min\left\{\frac{P_{00}Q_1}{P_{11}Q_0},\ 1 \right\},
\]
which is attained when 
\[
q_7 = P_{01},\ q_8 = \max\left\{0, \frac{P_{10}Q_1}{Q_0} - P_{01}  \right\}.
\]
For the lower bound, we solve 
\[
    \min\ \frac{P_{11} - q_8}{P_{11}} \quad \text{subject to}\quad 0\leq q_7\leq P_{01}, 0\leq q_8\leq P_{11}, \frac{P_{10}Q_1}{Q_0}\leq q_7+q_8\leq Q_1,
\]
but this minimum is equal to zero, because $q_7 = 0, q_8 = P_{11}$ is in the feasible set. Therefore, the sharp bounds on PN under MTS are given by 
\[
0\leq \textrm{PN} \leq \min\left\{ \frac{\Pr(Y=0\mid D=0)}{\Pr(Y=1\mid D=1)},\ 1  \right\}. 
\]
This completes the proof.  \qed

%In sum, the sharp bounds on PS, PN, and PNS under MTS without having MTR are given by 
%\begin{align*} 
%    0&\leq \textrm{PS}\leq \min\left\{ \frac{\Pr(Y=1\mid D=1)}{\Pr(Y=0\mid D=0)}, 1 \right\}, \\
%    0&\leq \textrm{PN} \leq \min\left\{ \frac{\Pr(Y=0\mid D=0)}{\Pr(Y=1\mid D=1)}, 1  \right\}, \\
%    0&\leq \textrm{PNS}\leq \min\left\{ \Pr(Y=0\mid D=0), \Pr(Y=1\mid D=1)   \right\}.
%\end{align*} 

\section{Comments on Confidence Intervals}\label{section:comments on inference}        
Confidence intervals in the case of interval identification have a well-known issues such as coverage rates for the parameter value versus the identified interval and their uniform validity: see e.g., \citet{Imbens/Manski:04,Stoye:07,canay2017practical}. Our inference problems (except the cases of PN and PS) are simpler because our situation is univariate in that the sharp lower bound is known to be zero, and it does not need to be estimated.  Below we explain this in more detail. 

To be specific, focus on $\theta$ with the sharp identified interval $[9, \theta_U]$ under MTR and MTS. The proposed confidence interval is 
\[
\mathrm{CI} := \left[ 0, \hat{\theta}_U + c_\alpha\cdot \widehat{\mathrm{SE}}_U\right],
\]
where $\hat{\theta}_U$ is an estimator of $\theta_U$ with the standard error equal to $\widehat{\mathrm{SE}}_U$, and $c_\alpha = \Phi^{-1}(1-\alpha)$ is the one-sided critical value from the standard normal distribution.  

This confidence interval covers the entire identified interval with probability approaching $1-\alpha$, and it covers the true parameter $\theta$ with probability at least $1-\alpha$. To see this point, note that 
\begin{equation}\label{eq:inf1}
\Pr\Bigl( [0,\theta_U]\subset \mathrm{CI} \Bigr)
=
\Pr\Bigl( \theta_U \leq \hat{\theta}_U + c_\alpha \cdot \widehat{\mathrm{SE}}_U\Bigr)
=
\Pr\Bigl( \frac{\hat{\theta}_U - \theta_U}{\widehat{\mathrm{SE}}_U} \geq -c_\alpha \Bigr),
\end{equation}
while letting $\theta = (1-a)\theta_U$ for some $a\in [0,1]$ shows that 
\begin{multline} \label{eq:inf2}
\Pr(\theta \in \mathrm{CI}) 
=
\Pr\Bigl( (1-a)\theta_U \leq \hat{\theta}_U + c_\alpha \widehat{\mathrm{SE}}_U \Bigr) 
\\
=
\Pr\Bigl( \frac{\hat{\theta}_U - \theta_U}{\widehat{\mathrm{SE}}_U} \geq -c_\alpha - a \frac{\theta_U}{\widehat{\mathrm{SE}}_U} \Bigr)
\geq    
\Pr\Bigl( \frac{\hat{\theta}_U - \theta_U}{\widehat{\mathrm{SE}}_U} \geq -c_\alpha \Bigr).
\end{multline}

In fact, if we let $\mathcal{P}:= \{ P:\ 0\leq \theta \leq \theta_U(P) \}$ be the set of data-generating processes under MTR and MTS, then the arguments in \Cref{eq:inf1,eq:inf2} show that for any distribution $P\in \mathcal{P}$, we must have
\begin{align*}
\Pr_P\Bigl( [0,\theta_U(P)]\subset \mathrm{CI} \Bigr) 
&= 
\Pr_P\left\{ \frac{\hat{\theta}_U - \theta_U(P)}{\widehat{\mathrm{SE}}_U} \geq -c_\alpha \right\}, \\
\Pr_P(\theta \in \mathrm{CI}) 
&\geq 
\Pr_P\left\{ \frac{\hat{\theta}_U - \theta_U(P)}{\widehat{\mathrm{SE}}_U} \geq -c_\alpha \right\},
\end{align*}
where $\Pr_P$ emphasizes that $P$ is the data-generating process under consideration. Therefore, the coverage properties of $\mathrm{CI}$ will be all uniformly valid over $P\in \mathcal{P}$ as long as the usual $t$-statistic converges uniformly to the standard normal in that  
\[
\liminf_{n\rightarrow \infty} \inf_{P\in \mathcal{P}} \Pr_P\left\{ \frac{\hat{\theta}_U - \theta_U(P)}{\widehat{\mathrm{SE}}_U}  \geq -c_\alpha \right\} = 1-\alpha.  
\]
\end{appendix}
\end{document}